\documentclass[letterpaper,USenglish]{lipics-v2016}
 
\usepackage{url}
\usepackage{hyperref}
\usepackage{amssymb}
\usepackage{amsmath}
\usepackage[inline]{enumitem}
\usepackage[ruled,vlined,linesnumbered]{algorithm2e}
\usepackage{graphicx,color}

\usepackage{mathtools}
\DeclarePairedDelimiter{\ceil}{\lceil}{\rceil}

\graphicspath{{figs/}}

\newcommand{\Reals}{{\mathbb{R}}}            
\newcommand{\eps}{\varepsilon}               
\newcommand{\argmin}{\operatornamewithlimits{argmin}}

\newcommand{\Frd}{Fr\'echet distance}
\newcommand{\FSD}{$\mathcal{FSD}_{\leq \eps}$}
\newcommand{\RD}{\rightarrow \mathbb{R}^d}
\newtheorem{observation}[theorem]{Observation}

\makeatletter
\def\footnoterule{\kern-3\p@
	\hrule \@width 2in \kern 2.6\p@} 
\makeatother

\def\H{{\cal H}}

\def\P{{\cal P}}

\newcommand{\defref}[1]{Definition~\ref{def:#1}}
\newcommand{\lemref}[1]{Lemma~\ref{lem:#1}}

\newcommand{\thmref}[1]{Theorem~\ref{thm:#1}}
\newcommand{\secref}[1]{Section~\ref{sec:#1}}
\newcommand{\subsecref}[1]{Section~\ref{subsec:#1}}
\newcommand{\subsubsecref}[1]{Section~\ref{subsubsec:#1}}
\newcommand{\algref}[1]{Algorithm~\ref{alg:#1}}
\newcommand{\alglineref}[1]{line~\ref{algline:#1}}
\newcommand{\obsref}[1]{Observation~\ref{obs:#1}}
\newcommand{\itemref}[1]{(\ref{#1})}

\def\LEP{\textsc{LongestEpsilonPrefix}}
\def\MinPrefix{\textsc{MinimumPrefix}}

\definecolor{lgray}{gray}{0.5}

\title{Fast Fr\'echet Distance Between Curves with Long Edges}
\titlerunning{Fast Fr\'echet Distance Between Curves with Long Edges} 

\author[1]{Joachim Gudmundsson}
\author[2]{Majid Mirzanezhad}
\author[3]{Ali Mohades}
\author[2]{Carola Wenk}
\affil[1]{University of Sydney, Australia, \texttt{joachim.gudmundsson@sydney.edu.au}}
\affil[2]{Tulane University, New Orleans, USA \texttt{\{mmirzane,cwenk\}@tulane.edu}}
\affil[3]{Amirkabir University of Technology, Tehran, Iran, \texttt{mohades@aut.ac.ir}}
  
\authorrunning{J.~Gudmundsson, M.~Mirzanezhad, A.~Mohades, and C.~Wenk} 

\Copyright{Joachim Gudmundsson, Majid Mirzanezhad, Ali Mohades, and 
Carola Wenk}


\EventEditors{John Q. Open and Joan R. Acces}
\EventNoEds{2}
\EventLongTitle{42nd Conference on Very Important Topics (CVIT 2016)}
\EventShortTitle{CVIT 2016}
\EventAcronym{CVIT}
\EventYear{2016}
\EventDate{December 24--27, 2016}
\EventLocation{Little Whinging, United Kingdom}
\EventLogo{}
\SeriesVolume{42}
\ArticleNo{23}

\begin{document}

		
\maketitle
	
\begin{abstract}
	Computing the~\Frd\ between two polygonal curves takes roughly quadratic time. In this paper, we show that for a special class of curves the Fr\'echet distance computations become easier. Let $P$ and $Q$ be two polygonal curves in $\Reals^d$ with $n$ and $m$ vertices, respectively. We prove four results for the case when all edges of both curves are long compared to the Fr\'echet distance between them: (1) a linear-time algorithm for deciding the Fr\'echet distance between two curves, (2) an algorithm that computes the \Frd\ in $O((n+m)\log (n+m))$ time, (3) a linear-time $\sqrt{d}$-approximation algorithm, and (4) a data structure that supports $O(m\log^2 n)$-time decision queries, where $m$ is the number of vertices of the query curve and $n$ the number of vertices of the preprocessed curve.  

\end{abstract}

\section{Introduction} \label{sec:introduction}	
%
Measuring the similarity between two curves is an important problem that has applications in many areas, e.g., in morphing~\cite{eghmm-nsmbp-02}, movement analysis~\cite{glw-mpstd-07}, handwriting recognition~\cite{skb-fdbas-07} and protein structure alignment~\cite{jxz-pssad-08}. Fr\'echet distance is one of the most popular similarity measures which has received considerable attentions in recent years. It is intuitively the minimum length of the leash that connects a man and dog walking across the curves without going backward. The classical algorithm for computing the Fr\'echet distance between curves with total complexity $n$ runs in $O(n^2 \log n)$ time~\cite{ag-cfdb-95}. The major goal of this paper is to focus on computing the Fr\'echet distance for a reasonable special class of curves in significantly faster than quadratic time. 
\subsection{Related Work}
Buchin et al.~\cite{bbkrw-hdwd-07} gave an $\Omega(n \log n)$ lower bound for computing the Fr\'echet distance. Then Bringmann~\cite{b-wwdtt-14} showed that, assuming the Strong Exponential Time Hypothesis, the Fr\'{e}chet distance cannot be computed in strongly subquadratic time, i.e., in time $O(n^{2-\epsilon})$ for any $\epsilon > 0$. For the \emph{discrete} Fr\'{e}chet distance, which considers only distances between the vertices, Agarwal et al.~\cite{aaks-cdfds-14} gave an algorithm with a (mildly) subquadratic running time of $O(n^2 \frac{\log\log n}{\log n})$. Buchin et al.~\cite{bbmw-fswd-14} showed that the continuous Fr\'echet distance can be computed in $O(n^2 \sqrt{\log n}(\log\log n)^{3/2})$ expected time. Bringmann and Mulzer~\cite{bw-adfd-15} gave an $O(n^2/\phi+n\log n)$-time algorithm to compute a $\phi$-approximation of the {discrete} Fr\'echet distance for any integer $1\leq \phi \leq n$. Therefore, an $n^\epsilon$-approximation, for any $\epsilon>0$, can be computed in (strongly) subquadratic time.
For the continuous Fr\'echet distance, there are also a few subquadratic algorithms known for restricted classes of curves such as $\kappa$-bounded, backbone and $c$-packed curves. Alt et al.~\cite{akw-cdmpc-04} considered $\kappa$-bounded curves and they gave an $O(n \log n)$ time algorithm to $(\kappa+1)$-approximate the Fr\'echet distance. A curve $P$ is $\kappa$-bounded if for any two points $x,y\in P$, the union of the balls with radii $r$ centered at $x$ and $y$ contains the whole $P[x,y]$ where $r$ is equal to $(\kappa/2)$ times the Euclidean distance between $x$ and $y$. For any $\epsilon>0$, Aronov et al.~\cite{ahkww-fdcr-06} provided a near-linear time $(1 + \epsilon)$-approximation algorithm for the discrete Fr\'echet distance for so-called backbone curves that have essentially constant edge length and require a minimum distance between non-consecutive vertices. For $c$-packed curves a $(1+\epsilon)$-approximation can be computed in $O(cn/\epsilon + cn \log n)$ time~\cite{dhw-afdrc-12}. A curve is $c$-packed if for any ball $B$, the length of the portion of $P$ contained in $B$ is at most $c$ times the diameter of $B$.

\subsection{Our Contribution}
In this paper, we study a new class of curves, namely curves with long edges, and we show that for these curves the Fr\'echet distance can be computed significantly faster than quadratic time.
In a particular application, one might be interested in detecting groups of different movement patterns in migratory birds that fly very long distances. As shown in Fig.~\ref{fig:motiv}, different flyways are comparatively straight and the trajectory data of individual birds usually consists of only one GPS sample per day in order to conserve battery power. Infrequent sampling and the straight flyways therefore result in curves with long edges, and it is desirable to compare the routes of different animals in order to identify common flyways. 

\begin{figure}[htbp]
	\begin{center}
		\includegraphics[width=5.5cm]{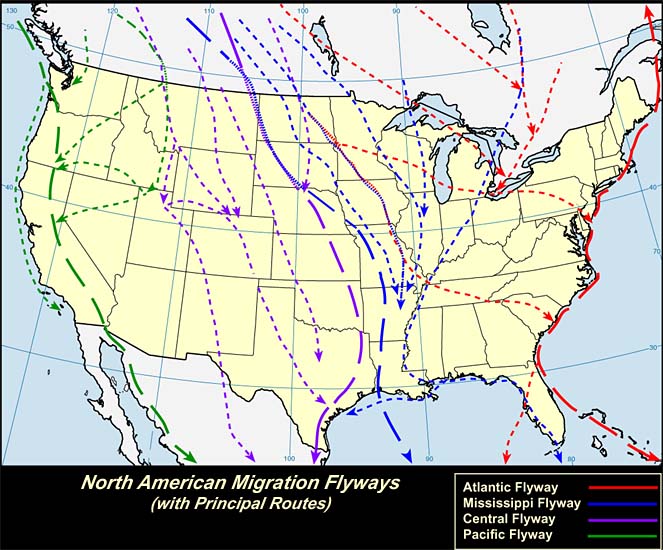}\caption{There are four typical flyways across the US. Clustering the trajectories by similarity between them allows us to detect the most common movement pattern~\protect\cite{namf}.}\label{fig:motiv}
	\end{center}
\end{figure}

We consider the decision, optimization, approximation and data structure problems for the \Frd~between two polygonal curves $P$ and $Q$  in $\Reals^d$ with $n$ and $m$ vertices, respectively, all for the case where all edges of both curves are long compared to the Fr\'echet distance between them. In \secref{greedyDecision} we present a greedy linear-time algorithm for deciding whether the~\Frd~is at most $\eps$, as long as all edges in $P$ are longer than $2\eps$ and edges in $Q$ are longer than $(1+\sqrt{d})\eps$. In \secref{optimization} we give an algorithm for computing the \Frd\ in $O((n+m)\log(n+m))$ time and a linear-time algorithm for approximating the \Frd\ up to a factor of $\sqrt{d}$. In \secref{datastructure} we present a data structure that decides whether the~\Frd~ between a preprocessed curve 	$P$ and a query curve $Q$ is at most $\eps$ or not, in $O(m\log^2 n)$ query time using $O(n\log n)$ space and preprocessing time.

%


\section{Preliminaries} \label{sec:preliminaries}
In this section we provide notations and definitions that will be required in the next sections. Let $P:[1,n]\rightarrow \Reals^d$ and $Q:[1,m]\rightarrow \Reals^d$ be two polygonal curves with vertices $p_1,\ldots, p_n$ and $q_1, \ldots, q_m$, respectively. We treat a \emph{polygonal curve} as a continuous map $P: [1,n] \rightarrow \mathbb{R}^d$ where $P(i)=p_i$ for an integer $i$, and the $i$-th edge is linearly parametrized as $P(i + \lambda) = (1-\lambda) p_i + \lambda p_{i +1}$, for integer $i$ and $0 < \lambda < 1$. A \emph{re-parametrization} $\sigma : [0, 1] \rightarrow [1, n]$ of $P$ is any continuous, non-decreasing function such that $\sigma(0) = 1$ and $\sigma(1) = n$. We denote a re-parametrization of $Q$ by $\theta:[0,1]\rightarrow [1,m]$. 
We denote the length of the shortest edge in $P$ and the length of the shortest edge in $Q$ by $l_P$ and $l_Q$, respectively. For two points $x,y\in \Reals^d$, let $\|x-y\|$ denote the Euclidean distance between the points and $\overline{xy}$ the straight line segment connecting $x$ to $y$.
The Euclidean distance between $x\in\Reals^d$ and an edge $e:[1,2] \rightarrow \Reals^d$ is denoted as $\|x,e\|=\min_{1\leq t\leq 2}\|x-e(t)\|$.
For $1\leq a \leq b \leq n$, $P[a,b]$ denotes the subcurve of $P$ starting in $P(a)$ and ending in $P(b)$.
Let $\eps>0$ be a real number. Consider an edge $e:[1,2] \rightarrow
\Reals^d$ of length $\|e\| > 2\eps$ whose endpoints are $e_1$ and
$e_2$. 
The \emph{direction vector} of $e$ is the vector from $e_1$ to $e_2$.
Now let $B(p,\eps)=\lbrace x\in \Reals^d \;|\; \|p-x\|\leq \eps \rbrace$ be
the ball with radius $\eps$ that is centered at a point $p$. The
\emph{cylinder} $C(e,\eps)$ is the set of points in $\Reals^d$ within
distance $\eps$ from $e$, i.e., $C(e,\eps)=\cup_{x \in e}
B(x,\eps)$. We say $P$ is \emph{$(e,\eps)$-monotone} if (1) $p_1 \in
B(e_1,\eps)$ and $p_n \in B(e_2,\eps)$, (2) $P \subseteq C(e,\eps)$,
and (3) $P$ is monotone with respect to the line supporting $e$. A
curve is monotone with respect to a line $l$ if it intersects any
hyperplane perpendicular to $l$ in at most one component.


\vspace{-2mm}


\begin{figure}[b!]
	\begin{center}
		\includegraphics[width=12.8 cm]{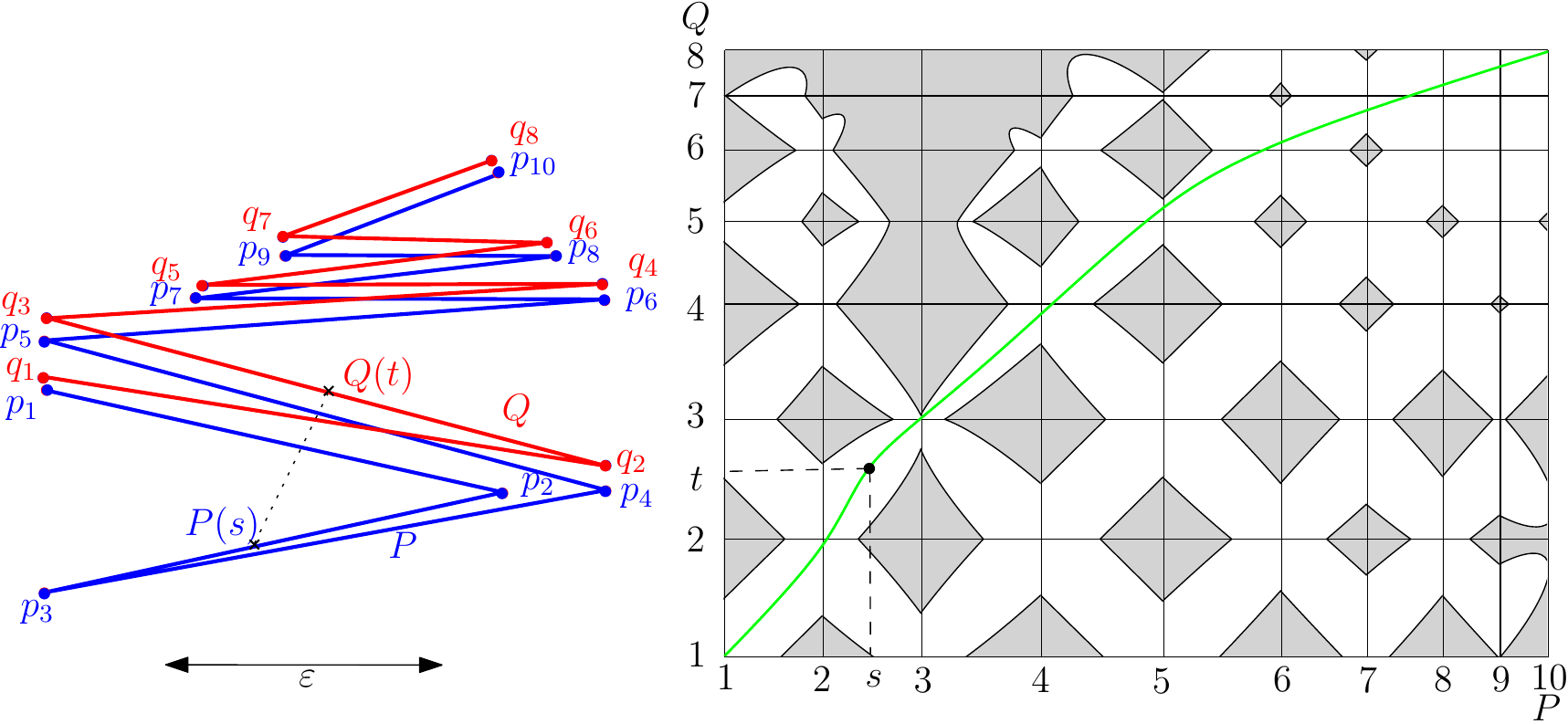}
		\caption{Two curves $P,Q$ and $\eps>0$ on the left, and the free space
			diagram \FSD$(P,Q)$ on the right showing free space in white and
			blocked space in gray. A reachable path is shown in green.
			The point $(s,t)$ lies in free space.
			There is a quadratic number of cells containing free
			space as well as a quadratic number of cells containing blocked
			space in \FSD$(P,Q)$ and all of them may need to be checked to
			decide reachability for $(n,m)$.
			Note that both $P$ and $Q$  contain short edges as well as long edges compared
			to $\eps$.
		}
		\label{fig:gcfs}
	\end{center}
\end{figure}

\subsection{Fr\'echet Distance and Free-Space Diagram}
To compute the Fr\'echet distance between $P$ and $Q$, Alt and
Godau~\cite{ag-cfdb-95} introduced the notion of \emph{free-space
	diagram}. For any $\eps>0$, we denote the free-space diagram between
$P$ and $Q$ by \FSD$(P,Q)$. This diagram has the domain $[1,n]\times[1,m]$ and it consists of $(n-1)\times (m-1)$ cells, where
each point $(s,t)$ in the diagram corresponds to two points $P(s)$ and
$Q(t)$.  A point $(s,t)$ in \FSD$(P,Q)$ is called \emph{free} if
$\|P(s)-Q(t)\|\leq \eps$ and \emph{blocked}, otherwise. The union of
all free points is referred to as the \emph{free space}. A
\emph{monotone matching} between $P$ and $Q$ is a pair of
re-parameterizations $(\sigma,\theta)$ corresponding to an
$xy$-monotone path from $(1,1)$ to $(n,m)$ within the free space in
\FSD$(P,Q)$. The \emph{Fr\'{e}chet distance} between two curves is
defined as $\delta_F(P, Q)= \inf_{(\sigma,\theta)} \max_{0\leq t \leq
	1} \| P(\sigma(t))-Q(\theta(t))\|$, where $(\sigma,\theta)$ is a
monotone matching and $\max_{0\leq t \leq 1}\|
P(\sigma(t))-Q(\theta(t))\|$ is called the \emph{width} of the
matching. A monotone matching realizing $\delta_F(P, Q)$ is called a
\emph{Fr\'echet matching}. A point $(s,t)$ is \emph{reachable} if
there exists a Fr\'echet matching from $(1,1)$ to $(s,t)$ in
\FSD$(P,Q)$. A Fr\'echet matching in \FSD$(P,Q)$ from $(1,1)$ to
$(s,t)$ is also called a \emph{reachable path} for $(s,t)$ (see
Fig.~\ref{fig:gcfs}). Alt and Godau~\cite{ag-cfdb-95} compute a reachable
path by propagating reachable points across free space cell boundaries
in a dynamic programming manner, which requires the exploration of the
entire \FSD$(P,Q)$ and takes $O(mn)$ time.

\subsection{The Main Idea}

We set out to provide faster algorithms for the Fr\'echet distance using implicit structural properties of the free-space diagram of curves with long edges. These properties allow us to develop greedy algorithms that construct valid re-parameterizations by repeatedly computing a maximally reachable subcurve on one of the curves. Like the greedy algorithm proposed by Bringmann and Mulzer~\cite{bw-adfd-15}, we compute prefix subcurves that have a valid Fr\'echet distance. However, while the approximation ratio of their greedy algorithm is exponential, the approximation ratio of the algorithm we present in~\subsecref{approximation} is constant, because we can take advantage of the curves having long edges. Our assumption on edge lengths is more general than backbone curves, since we do not require that non-consecutive vertices be far away from each other and we do not require any upper bound on the length of the edges. 

The free space diagram for curves with long edges is simpler, and intuitively seems to have fewer reachable paths (see Fig.~\ref{fig:lefs}). In the remainder of this paper we show that indeed we can exploit this simpler structure to compute reachable paths in a simple greedy manner which results in runtimes that are significantly faster than quadratic. 



\begin{figure}[htbp]
	\begin{center}
		\includegraphics[width=12.8 cm]{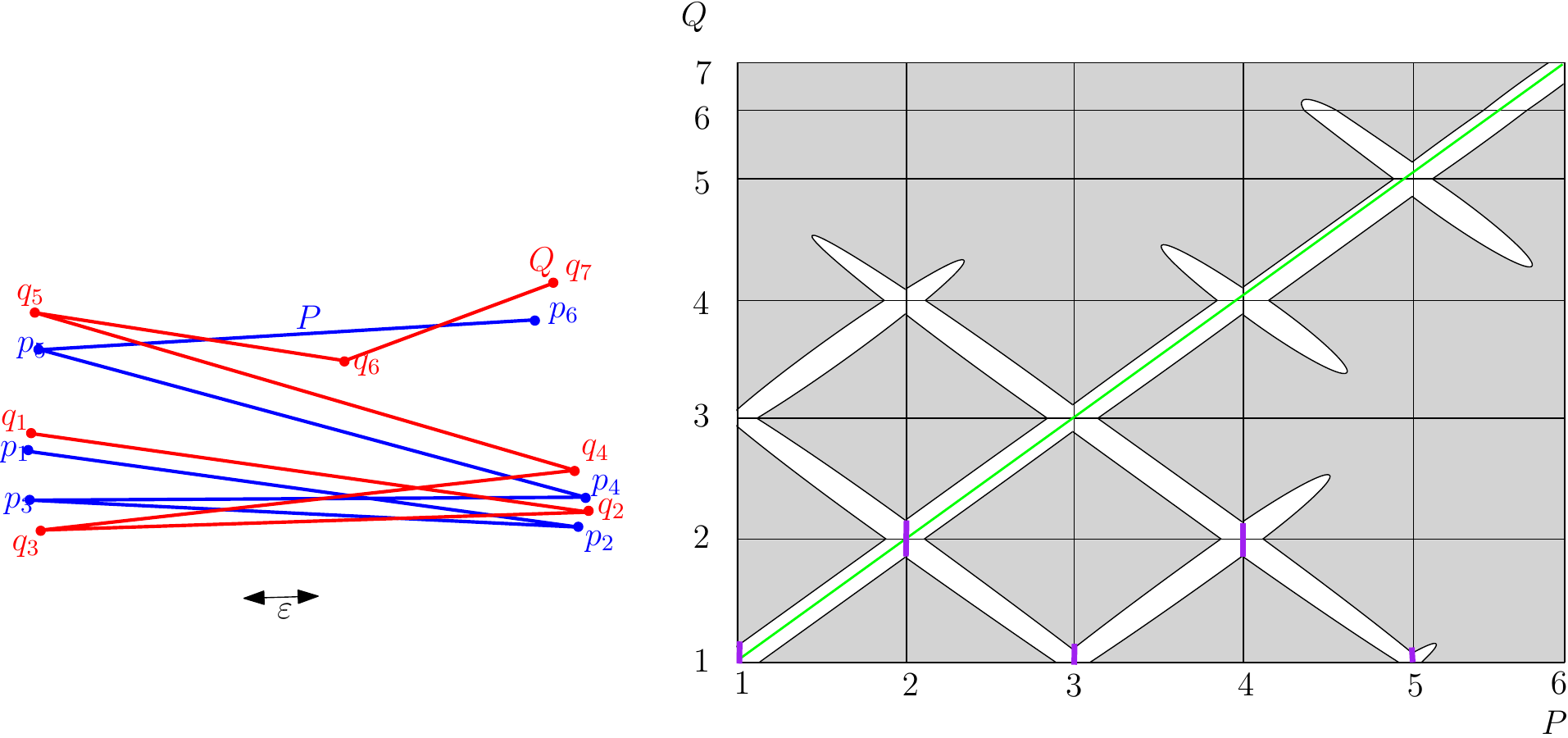}
		\caption{\FSD$(P,Q)$ for curves with long edges results in fewer
			reachable paths for $(n,m)$. 
			Consider the vertical free intervals (shown in purple) in the first row of
			the free space diagram. 
			Since
			$l_P>2\eps$, no consecutive purple intervals intersect which is a
			property we exploit. One can use such a property to find a reachable path without needing to check the entire free space diagram.
		}
		\label{fig:lefs}
	\end{center}
\end{figure}

\section{A Greedy Decision Algorithm}\label{sec:greedyDecision}

In this section we give a linear time algorithm for deciding whether the \Frd\ between two polygonal curves $P$ and $Q$ in $\Reals^d$ with relatively long edges is at most $\eps$. In~\subsecref{strprop}, we first prove a structural property for the case that each edge in $P$ is longer than $2\eps$ and $Q$ is a single segment. Afterwards in \subsecref{strextend}, we consider the extension to the case that $P$ and $Q$ are two polygonal curves and we show some extended structural property of free space induced by two curves with long edges. In \subsecref{thealg}, we present our greedy algorithm, which is based on computing longest reachable prefixes in $P$ with respect to each segment in $Q$. We consider three different variants of edge lengths assumption when $l_P>2\eps$ and $l_Q>(1+\sqrt{d})\eps$ (\subsubsecref{regularCase}), $l_P\geq2\eps$ and $l_Q\geq(1+\sqrt{d})\eps$ (\subsubsecref{geq}), and $l_P>0$ and $l_Q>4\eps$ (\subsubsecref{4eps}).  In~\subsecref{assumption}, we provide a critical example for which our greedy algorithm fails when the assumption on the edge lengths does not hold.

\subsection{A Simple Fr\'echet Matching for a Single Segment} \label{subsec:strprop}

In this section we start by introducing the crucial notion of \emph{orthogonal matching} between a polygonal curve $P$ and a single line segment $e$. An orthogonal matching projects each point from $P$ to its closest point on $e$. In particular, it maps vertices of $P$ either orthogonally to the segment $e$ or directly to the endpoints of $e$.

\begin{definition}[Orthogonal Matching] \label{def:orthomatch}
	Let $\eps>0$, $P:[1,n] \RD$ be a polygonal curve, and $e: [1,2] \RD $
	be a line segment. A Fr\'echet matching $(\sigma,\theta)$ realizing
	$\delta_F(P,e)\leq \eps$ is called an \emph{orthogonal matching} of
	width at most $\eps$ if
	$\sigma(t)=1$ for $t\in[0,a]$, 
	$\|P(\sigma(t))-e(\theta(t))\|=\|P(\sigma(t)),e\|\leq\eps$ for $t \in (a,b)$, and
	$\sigma(t)=n$ for $t\in[b,1]$
	for some $0\leq a\leq b\leq 1$; 
	see Fig.~\ref{fig:orthogonal}(a).
\end{definition}

\begin{figure}[htbp]
	\centering
	\includegraphics[width=11.5 cm]{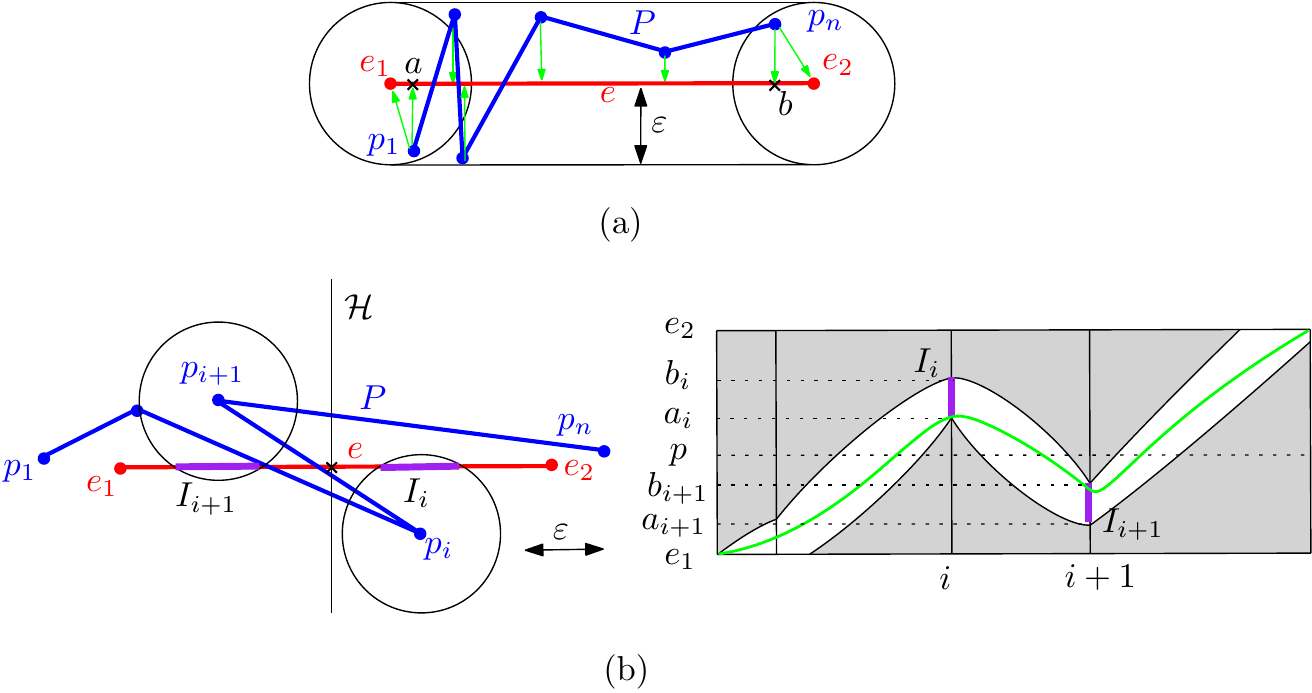}
	\caption{(a) In this example $P$ is $(e,\eps)$-monotone and the green arrows indicate an orthogonal matching between $P$ and $e$.
		(b) An illustration of the case $\H\cap e \neq\emptyset$ in the proof of \itemref{cond:1} $\Rightarrow$ \itemref{cond:2} in \lemref{monotonicity}. Note that the consecutive purple intervals $I_i$ and $I_{i+1}$ do not intersect because $l_P>2\eps$.
	}
	\label{fig:orthogonal}
\end{figure} 

Now we state a key lemma that demonstrates that if $P$ has long edges, then the orthogonal matching of width at most $\eps$ between $P$ and a segment $e$ exists if and only if $\delta_F(P,e)\leq \eps$, and this is equivalent to $P$ being $(e,\eps)$-monotone.

\begin{lemma}[Orthogonal Matching and Monotonicity]\label{lem:monotonicity}   
	Let $\eps >0$, $P:[1, n] \RD$ be a polygonal curve and $e:[1,2] \RD$
	be a line segment.
	Consider the following statements:	
	\begin{enumerate}
		\item $\delta_F(P,e)\leq\eps$, \label{cond:1}
		\item $P$ is $(e,\eps)$-monotone, \label{cond:2}
		\item $P$ and $e$ admit an \emph{orthogonal matching} of width at most
		$\eps$. \label{cond:3}
	\end{enumerate}
	
	\noindent
	In general, \itemref{cond:2}~$\Leftrightarrow$~\itemref{cond:3} and \itemref{cond:3}~$\Rightarrow$~\itemref{cond:1}.  
	\noindent
	In addition, if $l_P>2\eps$ then \itemref{cond:1}~$\Rightarrow$~\itemref{cond:2}, i.e., all three statements are equivalent.
\end{lemma}

\begin{proof}
	%
	We immediately have \itemref{cond:3}~$\Rightarrow$~\itemref{cond:1} by \defref{orthomatch}.
	To prove \itemref{cond:2}~$\Rightarrow$~\itemref{cond:3}, assume $P$
	is $(e,\eps)$-monotone. We can construct an orthogonal matching by
	mapping each $p_i$ to its nearest neighbor $e(1+t_i)$ on $e$, with $0\leq t_{i-1}\leq t_i\leq 1$.
	We set $\sigma(t_i)=i$ and $\theta(t_i)=1+t_i$ for all $i=1,\ldots,n$, and
	we set $a=t_1$, $b=t_n$, $\sigma(t)=1$ for $t\in[0,a]$, 
	$\sigma(t)=n$ for $t\in[b,1]$, and $\theta(0)=1$ and $\theta(1)=2$. The matching
	$(\sigma, \theta)$ is obtained by linearly interpolating between these values.
	The function $\sigma(t)$ is monotone by construction, and  $\theta(t)$ is monotone because $P$ is monotone with respect to the line supporting $e$. And all distances $\|P(\sigma(t))-e(\theta(t))\|\leq\eps$ because $P$ is $(e,\eps)$-monotone. Thus $(\sigma,\theta)$ is an orthogonal matching of width at most $\eps$.
	To prove \itemref{cond:3} $\Rightarrow$ \itemref{cond:2}, let $(\sigma,\theta)$ be an orthogonal matching of width at most $\eps$.
	Then clearly $p_1 \in B(e_1,\eps)$, $p_n \in B(e_2,\eps)$, and $P \subseteq C(e,\eps)$.
	Let $t_1,\ldots,t_n$ be such that $P(\sigma(t_i))=p_i$ for $i=1,\ldots,n$. Since $(\sigma,\theta)$ is a (monotone) Fr\'echet matching, $\theta(t_1), \ldots, \theta(t_n)$ is a
	monotone increasing sequence. And since $(\sigma,\theta)$ is orthogonal, the line segments $\overline{p_i\theta(t_i)}$ are all monotone to the line $\ell$ supporting $e$.  Therefore, $P$ is monotone with respect to $\ell$ and thus $P$ is $(e,\eps)$-monotone.

	Now assume $l_P> 2\eps$. In order to prove
	\itemref{cond:1}~$\Rightarrow$~\itemref{cond:2}, if $\delta_F(P,e)\leq
	\eps$ then clearly $p_1 \in B(e_1,\eps)$, $p_n\in B(e_2,\eps)$, and $P
	\subseteq C(e,\eps)$. It remains to show that $P$ is monotone with
	respect to the line $\ell$ supporting $e$.
	For all $i = 1,\ldots, n$, define $I_i=B(p_i,\eps)\cap e =
	e[a_i,b_i]$. Because $l_P> 2\eps$, we know that $I_i \cap
	I_{i+1}=\emptyset$. 
	Let $(\sigma,\theta)$ be a monotone matching realizing $\delta_F(P,e)\leq
	\eps$. For the sake of contradiction assume there exists a hyperplane
	$\H$ perpendicular to $\ell$ such that $P$ intersects $\H$ in at least
	two points $P(x)$ and $P(y)$, where $x < y$. Let $p_i$ be the last
	vertex along $P[x,y]$, and recall that $e_1$ and $e_2$ are the two
	vertices of $e$. First assume that $\H\cap e \neq \emptyset$. Then
	$p_i$ lies on the $e_2$-side of $\H$ and $p_{i+1}$ lies on the
	$e_1$-side of $\H$. Therefore, because $I_i\cap I_{i+1}=\emptyset$, we
	know that $a_i > b_{i+1}$. Let $t_i,t_{i+1} \in [0,1]$ be two values
	such that $p_i = P(\sigma(t_i))$ and $p_{i+1}=P(\sigma(t_{i+1}))$,
	where $t_i < t_{i+1}$. From $\sigma(t_i) \geq a_i $ and
	$\sigma(t_{i+1}) \leq b_{i+1}$, we know that $\sigma(t_i) >
	\sigma(t_{i+1})$, which violates the monotonicity of
	$(\sigma,\theta)$, see Fig.~\ref{fig:orthogonal}(b). Now consider the case
	that $\H\cap e=\emptyset$. Then $p_i$ lies on one side of $\H$, and
	$e$ lies entirely on the other side. If $\H \cap B(e_1,\eps) \neq
	\emptyset$, then we know that $P[1,y]\subseteq B(e_1,\eps)$. But this
	is not possible since all edges of $P$ are longer than $2\eps$. The
	same argument holds if $\H \cap B(e_2,\eps) \neq \emptyset$.

\end{proof}	
In fact \lemref{monotonicity} shows that for a curve $P$ with long edges, the \Frd~to a line segment $e$ is determined by examining whether $P$ is $(e,\eps)$-monotone or not.

\subsection{A Simple Fr\'echet Matching for More than One Segment} \label{subsec:strextend}
In this section, we extend the matching between a curve $P$ and a single line-segment $e$ to a matching between two curves $P$ and $Q$. 



\begin{definition}[Longest $\eps$-Prefix] \label{def:maximal}
	Let $\eps>0$, $P:[1,n]\RD$ be a polygonal curve, and $e:[1,2] \RD$ be a line segment.
	Define $\gamma=\max\{\,t\;|\; 1\leq t\leq n\;\mbox{and}\; \delta_F(P[1,t],e) \leq \eps\}$.
	We call $P[1,\gamma]$ the {\em longest $\eps$-prefix of $P$ with respect to $e$}.
\end{definition}

We now use the longest $\eps$-prefix to define an extension of the matching introduced in \defref{orthomatch}.
\noindent\defref{maximal} is the basis of our greedy algorithm (\algref{directedDecider}) which is presented in the next section. We show that if there exists a matching between two curves, then one can necessarily cut it into $m-1$ orthogonal matchings between each segment in $Q$ and the corresponding longest $\eps$-prefix. 
Before we reach this property, we need the following technical lemma:

\begin{lemma}[$(\sqrt{d}\eps)$-Ball] \label{lem:fp} Let $\eps>0$ and let $P:[1,n]\RD$ be a polygonal curve such that $l_P>2\eps$. Let $e:[1,2]\RD$ where $\|e\|> 2\eps$. Assume that $P[1,\gamma]$ is the longest $\eps$-prefix of $P$ with respect to $e$, and let $\alpha$ be a parameter such that $P(\alpha)$ is the first point along $P$ that intersects $B(e_2,\eps)$. Then $P[\alpha,\gamma] \subseteq B(e_2,\sqrt{d}\eps)$.
\end{lemma}

\begin{proof}
	By assumption $\|e\|>  2\eps$, we know that $B(e_1,\eps) \cap B(e_2,\eps)=\emptyset$, thus $\alpha$ exists. Notice that $P[\alpha,\gamma]\subseteq C(e,\eps)$. Let $\H$ be the hyperplane that is  intersecting and perpendicular to $e$ and is tangent to $B(e_2,\eps)$. Hence $\H$ splits $P[1,\gamma]$ into two parts, the part on the $e_1$-side and the part that on the $e_2$-side. Let $P(x)$ be the last vertex before $P(\gamma)$ along $P$. By~\defref{maximal}, $\delta_F(P[1,\gamma],e)\leq \eps$, and (1) if $P(x)\in B(e_2,\eps)$, then~\lemref{monotonicity} implies that $P[1,x]$ is $(e,\eps)$-monotone. Thus $P[\alpha,\gamma]$ must lie on the $e_2$-side of $\H_2$, and in particular inside the cube enclosing $B(e_2,\eps)$, see Fig.~\ref{fig:sqrtd}. Therefore the maximum possible distance between any point in $P[\alpha,\gamma]$ and $e_2$ is $\sqrt{d}\eps$. 
	(2) If $P(x)\notin B(e_2,\eps)$, we first show that $P[1,x]$ is monotone with respect to the line supporting $e$ and then we use the similar argument as in (1) to imply the maximum possible distance between any point in $P[\alpha,\gamma]$ and $e_2$ is $\sqrt{d}\eps$. Now let $(\sigma,\theta)$ be a Fr\'echet matching between $P[1,\gamma]$ and $e$. For the sake of contradiction assume there exists an edge $P[i,i+1]$ such that the angle between the direction vectors of $P[i,i+1]$ and $e$ is greater than $\pi/2$ with $i< x$. Let $t_i,t_{i+1} \in [0,1]$ be two real values with $t_{i}< t_{i+1}$ such that $\sigma(t_i)=i$ and $\sigma(t_{i+1})=i+1$ and let $I_i=B(p_i,\eps)\cap e=e[a_i,b_i]$ and $I_{i+1}=B(p_{i+1},\eps)\cap e=e[a_{i+1},b_{i+1}]$. Now from $B(p_i,\eps)\cap B(p_{i+1},\eps)=\emptyset$ follows that $I_i\cap I_{i+1}=\emptyset$. Note that the angle between the direction vectors of  $P[i,i+1]$ and $e$ is greater than $\pi/2$ which indicates that $b_{i+1}<a_i$. Therefore $a_{i+1}\leq\theta(t_{i+1})\leq b_{i+1}< a_i \leq \theta(t_i)\leq b_i$. Now three following cases are expected:
	\begin{itemize*}
		\item[(i)]if $i+1<\alpha$, then $\gamma$ does not exist since $(\sigma,\theta)$ is not monotone and this would be a contradiction. Therefore $P[1,x]$ is monotone with respect to the line supporting $e$.
		
		\item[(ii)]  If $\alpha< i \leq x$, then $\gamma<x$ since $i<\gamma<i+1$ which is a contradiction with $\gamma>x$. Hence $P[1,x]$ is monotone with respect to the line supporting $e$.
		
		\item[(iii)] if $i= x\leq \alpha$, then $P[\alpha,\gamma]$ is only a subsegment of $P[i,i+1]$ and trivially lies within $B(e_2,\eps)$.
		
	\end{itemize*}  
	This completes the proof.
	
\end{proof}

\begin{figure} [htbp]
	\begin{center} 
		\includegraphics[width=6.5 cm]{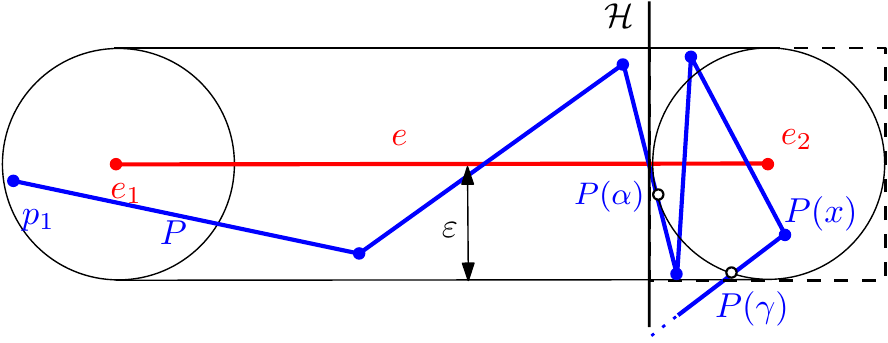}
		\caption{The farthest point in $P[\alpha,\gamma]$ from $e_2$ must lie inside the cube enclosing $B(e_2,\eps)$.}
		\label{fig:sqrtd}
	\end{center}
\end{figure}

\begin{lemma} [($3\eps)$-Ball]\label{lem:3delta-ball}
	Let $\eps>0$ and let $P:[1,n]\RD$ be a polygonal curve. Let $e:[1,2]\RD$ where $\|e\|> 2\eps$. Assume that $P[1,\gamma]$ is the longest $\eps$-prefix of $P$ with respect to $e$ and $P(\alpha)$ is the first point along $P$ that intersects $B(e_2,\eps)$. Then $P[\alpha,\gamma]\subseteq B(e_2,3\eps)$.
\end{lemma}

\begin{proof}
	Although the proof of Lemma 11 in Gudmundsson and Smid~\cite{gs-fqgt-15} is similar, we describe a slight modification of the proof that is necessary for our setting. Suppose $(\sigma,\theta)$ is a Fr\'echet matching realizing $\delta_F(P[1,\gamma],e)\leq \eps$. Let $x \in [\alpha,\gamma]$ such that $P(x)$ is the farthest point to $e_2$. We need to show that $\|P(x)-e_2\|\leq 3\eps$ which implies $P[\alpha,\gamma] \subseteq B(e_2,3\eps)$. Let $t_\alpha, t_\gamma \in[0,1]$ be two values such that $\alpha=\sigma(t_\alpha)$ and $\gamma=\sigma(t_\gamma)$. Note that there exists some $t_x\in [t_\alpha,t_\gamma]$ such that $x=\sigma(t_x)$. 
	By the triangle inequality we have: $$\|P(x)-e_2\|\leq \|P(x)-e(\theta(t_x))\|+\|e(\theta(t_x))-e_2\|\leq \eps+ \|e(\theta(t_x))-e_2\|.$$ Note that $t_x > t_\alpha$ and we can have $\|e(\theta(t_x))-e_2\| \leq \|e(\theta(t_\alpha))-e_2\|$, hence: $$\|P(x)-e_2\| \leq \eps + \|e(\theta(t_\alpha))-e_2\|.$$ 
	By applying the triangle inequality once more we have:
	$$\|P(x)-e_2\| \leq \eps+  \|e(\theta(t_\alpha))-P(\alpha)\| + \|P(\alpha)-e_2\| \leq 3\eps.$$
\end{proof}

Now we show that if $\delta_F(P,Q)\leq \eps$, then the two polygonal curves $P$ and $Q$ admit a piecewise orthogonal matching, which can be obtained by computing longest $\eps$-prefixes of $P$ with respect to each segment of $Q$. This lemma is the foundation of our greedy algorithm (\algref{directedDecider}).

\begin{figure} [htbp]
	\begin{center} 
		\includegraphics[width=12.5 cm]{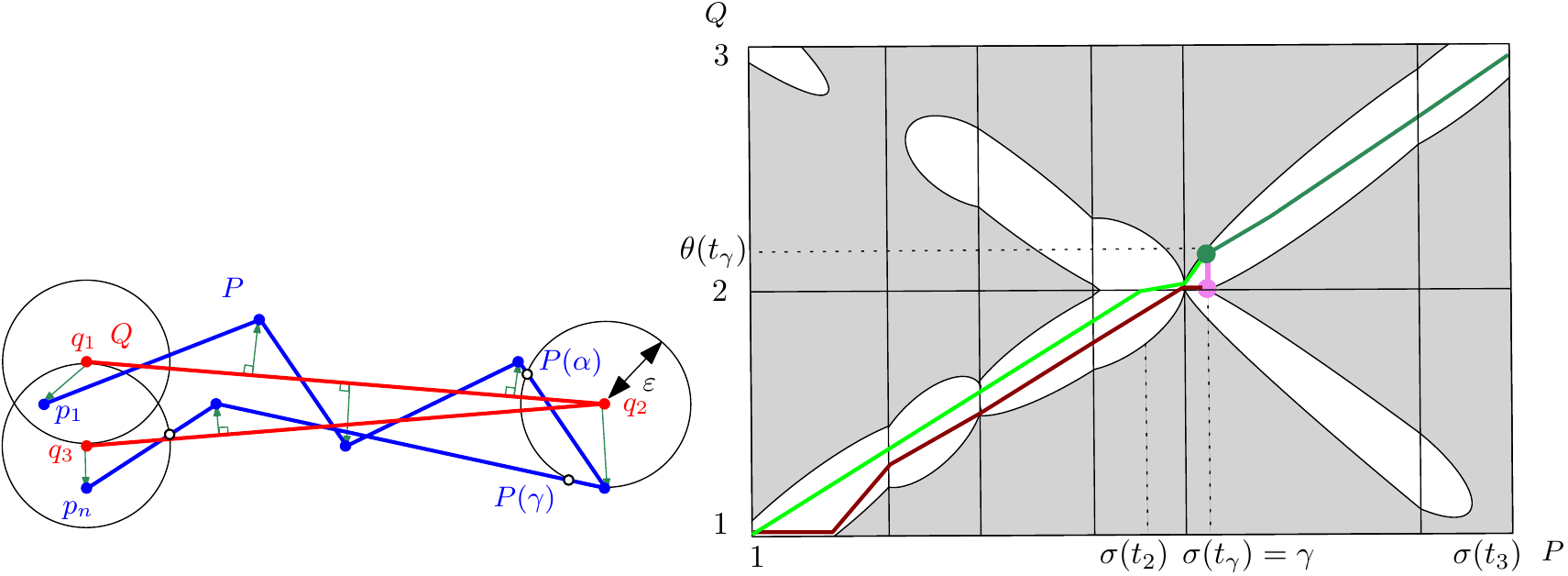}
		\caption{ Given an arbitrary matching (the concatenation of the light and dark green reachable paths), the orthogonal matching (the brown reachable path) between $P[1,\gamma]$ and $Q[1,2]$ exists. We construct a matching realizing $\delta_F(P[\gamma,n],Q[2,m])\leq \eps$ as the concatenation of the pink and the dark green reachable paths.} 
		\label{fig:piecewise}
	\end{center}
\end{figure}


\begin{lemma}[Cutting Lemma] \label{lem:pwmatching}
	Let $\eps >0$, and let $P:[1,n]\RD$ and $Q:[1,m]\RD$ be two polygonal curves such that $l_P>2\eps$ and $l_Q>(1+\sqrt{d})\eps$. If $\delta_F(P,Q)\leq \eps$, then $P[1,\gamma]$ as the longest $\eps$-prefix of $P$ with respect to $Q[1,2]$ exists, $\delta_F(P[1,\gamma], Q[1,2])\leq \eps$ and $\delta_F(P[\gamma, n], Q[2,m])\leq \eps$.
\end{lemma}

\begin{proof}
	Let $(\sigma,\theta)$ be any Fr\'echet matching realizing
	$\delta_F(P,Q)\leq \eps$. This corresponds to a reachable path, which
	is shown as the concatenation of the light and dark green paths in the example
	in Fig.~\ref{fig:piecewise}. Let $t_2\in[0,1]$ be the largest value such
	that $Q({\theta(t_2)})=q_2$, hence
	$\delta_F(P[1,\sigma(t_2)],Q[1,2])\leq \eps$. By \defref{maximal},
	$\gamma$ exists with $\gamma \geq \sigma(t_2)$, and
	$\delta_F(P[1,\gamma], Q[1,2])\leq \eps$. See the brown reachable path
	corresponding to the orthogonal matching realizing
	$\delta_F(P[1,\gamma], Q[1,2])\leq \eps$ in Fig.~\ref{fig:piecewise}. In the
	remainder of this proof we construct a matching to prove that
	$\delta_F(P[\gamma, n], Q[2,m])\leq \eps$ (the concatenation of the pink
	and dark green paths).
	
	Let $t_\gamma \in[0,1]$ be the largest value such that
	$P(\sigma(t_\gamma))=P(\gamma)$. By \lemref{fp},
	$P[\sigma(t_2),\gamma] \subseteq B(q_2,\sqrt{d}\eps)$.  Now let
	$t_3\in[0,1]$ be the smallest value such that $Q(\theta(t_3))=q_3$. We
	have $\|q_2-q_3\|>(1+\sqrt{d})\eps$, therefore
	$B(q_2,\sqrt{d}\eps)\cap B(q_3,\eps)=\emptyset$ and thus
	$(\sigma,\theta)$ cannot match $q_3$ to any point in
	$P[\sigma(t_2),\gamma]$. Therefore, $\sigma(t_2)\leq
	\gamma=\sigma(t_\gamma) <\sigma(t_3)$, and correspondingly
	$\theta(t_2) \leq \theta(t_\gamma) < \theta(t_3)$.
	
	Now we construct a new matching $(\bar{\sigma},\bar{\theta})$
	realizing $\delta_F(P[\gamma, n], Q[2,m])\leq \eps$ as follows:
	$\bar{\sigma}(t)=\sigma(t)$ and $\bar{\theta}(t)=\theta(t)$ for all
	$t_\gamma\leq t \leq 1$ (dark green reachable path). On the other
	hand, since $\|P(\gamma) - q_2\|\leq \eps$ (pink point) and
	$\|P(\gamma) - Q({\theta(t_\gamma)})\| \leq \eps$ (dark green
	point), we know that $Q[2,\theta(t_\gamma)]\subseteq
	B(P(\gamma),\eps)$, i.e., the pink vertical segment is free. We set,
	$\bar{\sigma}(t)=\gamma$ and
	$\bar{\theta}(t)=\frac{t_\gamma-t}{t_\gamma}\cdot 2
	+\frac{t}{t_\gamma}\cdot \theta(t_\gamma)$ for all $t_2 \leq t \leq
	t_\gamma$ (pink reachable path). Therefore, we have
	$\delta_F(P[\gamma, n], Q[2,m])\leq \eps$, which completes the proof.
\end{proof}

Now since by~\lemref{pwmatching} we have $\delta_F(P[1,\gamma],Q[1,2])\leq \eps$, ~\lemref{monotonicity} implies that the matching between $P[1,\gamma]$ and $Q[1,2]$ is orthogonal. 
Let $P(x)$ be the last vertex of $P[1,\gamma]$ and let $Q(x')$ be its closest point on $Q[1,2]$, for some $x<\gamma$ and $x'\leq 2$. Note that if $\|P(\gamma)-P(x)\|$ is shorter than $2\eps$, we can adjust the orthogonal matching by simply mapping all points on $P[x,\gamma]$ to $Q[x',2]$.
In addition, if $P$ and $Q$ have long edges then the free-space diagram is simpler than in the general case, 
since the entire vertical space (the pink segment in Fig.~\ref{fig:piecewise}) between the two points $(\gamma,2)$ and $(\gamma,\theta(t_\gamma))$ has to be free and cannot contain any blocked points. 


\subsection{The Decision Algorithm} \label{subsec:thealg}
In this section we present a linear time decision algorithm using the
properties provided in~\subsecref{strprop} and \subsecref{strextend}. In \subsubsecref{regularCase} we consider the case that $l_P>2\eps$ and $l_Q>(1+\sqrt{d})\eps$. In \subsubsecref{geq} we show that this approach can be generalized to the case that $l_P\geq 2\eps$ and $l_Q\geq(1+\sqrt{d})\eps$, and in \subsubsecref{4eps} we generalize the approach to the case that there is only an edge length assumption on $Q$.

\subsubsection{Long Edges with $l_P>2\eps$ and $l_Q>(1+\sqrt{d})\eps$} \label{subsubsec:regularCase}

At the heart of our decision algorithm is the greedy algorithm presented in
\algref{directedDecider}. The input to this \textsc{DecisionAlgorithm}
are two polygonal curves $P$ and $Q$, and $\eps>0$. The algorithm
assumes that $P$ and $Q$ have long edges. In each iteration the
function \LEP~returns $\gamma$, where $P[s,\gamma]$ is the longest
$\eps$-prefix of $P[s,n]$ with respect to $Q[i-1,i]$, if it
exists. Here, $s$ is the parameter where $P(s)$ is the endpoint of
the previous longest $\eps$-prefix with respect to $Q[i-2,i-1]$. At
any time in the algorithm, if $\gamma=null$, this means that the
corresponding longest $\eps$-prefix does not exist and then ``No'' is
returned. Otherwise, the next edge of $Q$ is processed. This continues
iteratively until all edges have been processed, or $\gamma_i$ does not exist for some $i=2,\cdots,m$.

\begin{algorithm}[h] 
	\DontPrintSemicolon
	\SetKwFunction{DirectedDecider}{\textsc{DirectedDecider}}
	\SetKwFunction{DecisionAlgorithm}{\textsc{DecisionAlgorithm}}
	\SetKwProg{myalg}{}{}{}
	\SetKwFunction{LastIntersectionPoint}{\textsc{LastIntersectionPoint}}
	\SetKwFunction{FirstIntersectionPoint}{\textsc{FirstIntersectionPoint}}
	\SetKwFunction{MaximalPoint}{\textsc{MaximalPoint}}
	\SetKwFunction{LongestEpsilonPrefix}{\LEP}
	%
	\caption{Decide whether $\delta_F(P,Q)\leq\eps$}
	\label{alg:directedDecider}	
	
	\BlankLine
	\myalg{\DecisionAlgorithm{$P[1, n], Q[1, m], \eps$}}{
		\hspace{-1mm}\tcp*[l]{Assumes $l_P>2\eps~\mbox{and}~l_Q>(1+\sqrt{d})\eps$
		}
		$\gamma_1 \leftarrow 1$\;  
		\For{$i \leftarrow 2$ \KwTo $m$}
		{
			$\gamma_i \leftarrow$\LongestEpsilonPrefix{$P[\gamma_{i-1},n], Q[i-1,i],\eps$}\; 
			\lIf{$\gamma = null$}
			{
				\KwRet{``No''}
			}
			{
				$s \leftarrow \gamma$\;
			}
		}
		\lIf{$\gamma <n$}
		{\KwRet{``No''}}
		{\KwRet{``Yes''}}		
	}
\end{algorithm}

%
%
The $\LEP(P[\gamma_{i-1},n],Q[i-1,i],\eps)$ procedure is implemented as follows:
We use Alt and Godau's~\cite{ag-cfdb-95}  dynamic programming algorithm to
compute the reachability information in \FSD$(P[\gamma_{i-1},n],Q[i-1,i])$, 
which computes all $(s,t)$ for which $\delta_F(P[\gamma_{i-1},s],Q[i-1,t])\leq\eps$.
This takes linear time in the complexity of $P[\gamma_{i-1},n]$ since $Q[i-1,i]$ is a single
segment. 
Then $\gamma_i$ is the largest $s$ for which $\delta_F(P[\gamma_{i-1},s],Q[i-1,i])\leq\eps$.
Note that $P(s)$ has to lie on the boundary of $B(q_i,\eps)$.
If no such $s$ exists then $\gamma_i=null$.
%
%
We now prove the correctness of our decision algorithm.

\begin{theorem}[Correctness] \label{thm:correctness1}
	Let $\eps >0$, and let $P:[1,n]\RD$ and $Q:[1,m]\RD$ be two polygonal curves such that $l_P>2\eps$ and $l_Q>(1+\sqrt{d})\eps$.	Then \textsc{DecisionAlgorithm}($P,Q,\eps$) returns ``Yes'' if and only if 
	$\delta_F(P,Q)\leq \eps.$
\end{theorem}

\begin{proof}
	If the algorithm returns ``Yes'' then the sequence $\lbrace
	(q_i,\gamma_i) \rbrace$ for all $i=1,\ldots, m$ with $\gamma_1=1$
	and $\gamma_m=n$ describes a monotone matching that realizes
	$\delta_F(P,Q)\leq \eps$.
	
	If $\delta_F(P,Q) \leq \eps$, then we prove
	by induction on $i$ that the algorithm returns ``Yes'', i.e., all
	longest $\eps$-prefixes $(P[1,\gamma_2], P[\gamma_2,\gamma_3], \ldots
	,P[\gamma_{m-1},\gamma_m])$ of $P$ with respect to the corresponding
	segments of $Q$ exist. For $i=2$, following \lemref{pwmatching},
	$\gamma_2$ exists and can be found by the algorithm.  For any $i>2$,
	the algorithm has determined $\gamma_2,\ldots, \gamma_{i-1}$ already
	and by \lemref{pwmatching}, $\delta_F(P[\gamma_{i-1}, n],Q[i-1, m])
	\leq \eps$. Another application of \lemref{pwmatching} yields that
	$\delta_F(P[\gamma_{i-1}, \gamma_i],Q[i-1, i]) \leq \eps$ and
	$\delta_F(P[\gamma_i, n],Q[i, m]) \leq \eps$.

	In the case that $i=m-1$ it remains to prove that $\gamma_{i+1}=\gamma_m=n$. For the sake of contradiction, assume $\gamma_m < n$. Since $P[\gamma_{m-1},\gamma_m]$ is the longest $\eps$-prefix, there is no other $\gamma'_m \in (\gamma_m,n]$ such that   $\delta_F(P[\gamma_{m-1},\gamma'_m],Q[m-1,m]) \leq \eps$. Consequently, $\delta_F(P[\gamma_{m-1},\gamma'_m],Q[m-1,m]) > \eps$ and therefore $\delta_F(P[\gamma_{m-1},n],Q[m-1,m]) > \eps$. Applying the contrapositive of \lemref{pwmatching} to $P[\gamma_{m-1},n]$ and $Q[m-1,m]$ yields $\delta_F(P,Q) > \eps$, which is a contradiction. Therefore $\gamma_m=n$ and the algorithm returns ``Yes'' as claimed.
\end{proof}




\begin{observation}[Piecewise Orthogonal Matching] \label{obs:piecewise}
	If $\delta_F(P,Q)\leq \eps$, then the sequence $\lbrace
	\gamma_{1},\gamma_2,\ldots, \gamma_n\rbrace$ computed by
	\algref{directedDecider} induces a Fr\'echet matching that maps
	$P(\gamma_{i})$ to $q_i$, and therefore $\delta_F(P,Q)\leq\eps$ for all
	$i=2,\ldots, m$. ~\lemref{monotonicity} implies that the matching between
	$P[\gamma_{i-1},\gamma_{i}]$ and $Q[i-1,i]$ is orthogonal.
\end{observation}


We summarize this section with the following theorem:
\begin{theorem}[Runtime] \label{thm:greedyalg}
	Let $\eps >0$, and let $P:[1,n]\RD$ and $Q:[1,m]\RD$ be two polygonal curves such that $l_P>2\eps$ and $l_Q>(1+\sqrt{d})\eps$. Then there exists a greedy decision algorithm, \algref{directedDecider},
	that can determine whether $\delta_F(P,Q) \leq \eps$ in $O(n+m)$ time.
\end{theorem}

\begin{proof}
	The number of vertices in $P[\gamma_{i-1},\gamma_{i}]$ is at most $\ceil{\gamma_{i}-\gamma_{i-1}}+1$. The algorithm greedily finds the longest $\eps$-prefix per edge $Q[i-1,i]$ by calling \LEP$(P[s,n], Q[i-1,i],\eps)$ in $O(\ceil{\gamma_{i}-\gamma_{i-1}}+1)$ time. The for-loop iterates over $m-1$ edges, thus the runtime is $\sum_{i=2}^{m} (\ceil{\gamma_{i}-\gamma_{i-1}}+1) < \sum_{i=2}^{m}(\gamma_{i}-\gamma_{i-1}+2)=\gamma_m-\gamma_{1}+2(m-1)= n-1+2m-2 = O(n+m)$.
	
\end{proof}

\subsubsection{Long Edges with $l_P\geq 2\eps$ and $l_Q\geq(1+\sqrt{d})\eps$} \label{subsubsec:geq}

We now consider the slightly more general case that $l_P\geq 2\eps$ and $l_Q\geq(1+\sqrt{d})\eps$. The optimization
algorithm presented in \subsecref{opt} makes use of this case.
Clearly, if $l_P> 2\eps$ and $l_Q>(1+\sqrt{d})\eps$ then \thmref{correctness1} applies as usual.  
If $l_P=2\eps$ or $l_Q=(1+\sqrt{d})\eps$ then \algref{directedDecider} can still
be run, however  the Fr\'echet matching induced by the
$\gamma_i$ is not necessarily a piecewise orthogonal matching anymore,
which means \obsref{piecewise} may not hold, see
Fig.~\ref{fig:twoepsmonotone}. 
However, we can still prove a slightly modified correctness theorem.
%


\begin{figure} 
	\centering
	\includegraphics[width=.6\textheight]{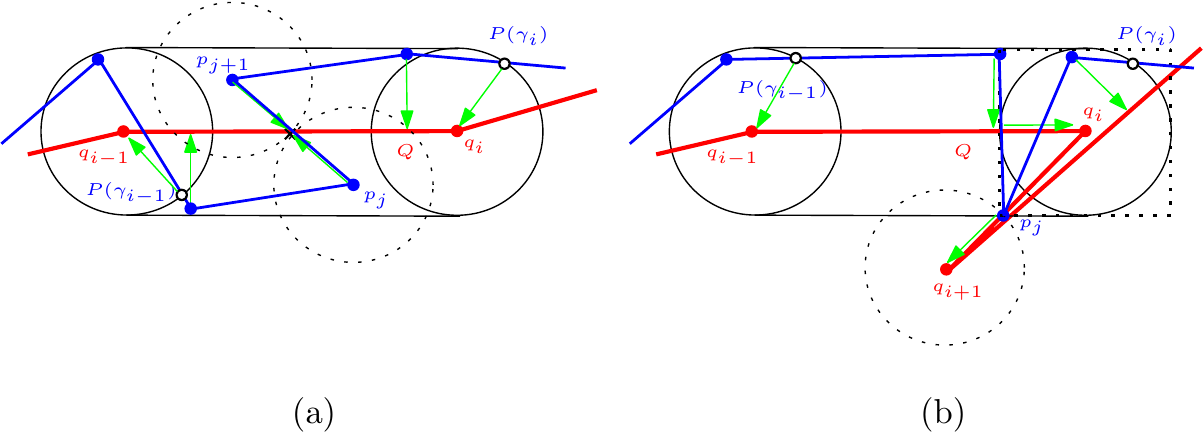}
	\caption{Two examples of matchings  between $P[\gamma_{i-1},\gamma_i]$ and $Q[i-1,i]$ that are not (piecewise) orthogonal. Matchings are indicated with green arrows. All balls have radius $\eps$.
		(a) An example where $l_P=\|p_{j+1}-p_j\|= 2\eps$
		and $\gamma_i$ exists, but the induced Fr\'echet matching is not orthogonal.
		(b) 
		An example where $\|q_{i+1}-q_i\|=(1+\sqrt{d})\eps$ and
		$B(q_{i+1},\eps)\cap B(q_i,\sqrt{d}\eps)\cap P= p_j$. Although $\gamma_{i}$ exists, a matching that is not piecewise orthogonal of width exactly $\eps$ exists.}
	\label{fig:twoepsmonotone}
\end{figure}


\begin{theorem} \label{thm:equal2eps}
	Let $\eps>0$, and
	$l_P\geq 2\eps$ and $l_Q\geq (1+\sqrt{d})\eps$.
	If
	\textsc{DecisionAlgorithm$(P,Q,\eps)$}
	returns ``Yes'' then
	$\delta_F(P,Q)\leq \eps$. If it returns ``No'' then $\delta_F(P,Q)\geq \eps$.
\end{theorem}
\begin{proof}
	Let $\eps^*= \delta_F(P,Q)$.
	If $l_P> 2\eps$ and $l_Q>(1+\sqrt{d})\eps$ then \thmref{correctness1} applies as usual.  
	So, assume $l_P=2\eps$ or $l_Q=(1+\sqrt{d})\eps$.
	If the algorithm returns ``Yes'', then  we know that $\delta_F(P[\gamma_{i-1},\gamma_i],Q[i-1,i])\leq\eps_0$ for all $i=2,\ldots,m$, and therefore $\eps^*\leq\eps$.

	In the remainder of this proof we show the contrapositive of the
	second part: If $\eps^*=\delta_F(P,Q)<\eps$ then
	\textsc{DecisionAlgorithm$(P,Q,\eps)$} returns ``Yes''.
	%
	%
	So, assume $\eps^*<\eps$. Then, by \thmref{correctness1}, 
	\textsc{DecisionAlgorithm$(P,Q,\eps^*)$} returns ``Yes'',
	which means that all
	$\gamma^*_i=\LEP(P[\gamma^*_{i-1},n],Q[i-1,i],\eps^*)$ exist for
	all $i=2,\ldots,m$, and  $\gamma^*_1=1$. We prove by induction that all $\gamma_i=\LEP(P[\gamma_{i-1},n],Q[i-1,i],\eps)$ exist as well. The inductive base is trivial to show since $\gamma_1=\gamma^*_1=1$. Now as an inductive hypothesis let $i>1$ be the largest integer value for which $\gamma_{i-1}$ exists and is computed.  
	In the following we show that $\gamma_i=\LEP(P[\gamma_{i-1},n],Q[i-1,i],\eps)$ exists and can be computed.
	Let $P(x)$ be the first point along $P[\gamma^*_i,n]$ on the boundary of
	$B(q_i,\eps)$. We have $\gamma^*_{i-1}<\gamma_{i-1}<\gamma^*_i<x$, where the first inequality
	follows from $B(q_{i-1},\eps^*)\subset B(q_{i-1},\eps)$, and the second inequality follows from
	$B(q_{i-1},\eps)\cap B(q_i,\eps^*)=\emptyset$ because $l_Q>2\eps$.
	Now let $(\sigma,\theta)$ be the Fr\'echet matching realizing
	$\delta_F(P[\gamma^*_{i-1},\gamma^*_i],Q[i-1,i])\leq \eps^*$, and
	let $t\in [0,1]$ such that $\sigma(t)=\gamma_{i-1}$.
	Then from $\gamma^*_{i-1}< \gamma_{i-1}<\gamma^*_i$ follows that
	$i-1 \leq \theta(t)\leq i$.
	We can therefore construct a piecewise re-parameterization for $P[\gamma_{i-1},x]$ and $Q[i-1,i]$ which yields:
	%
	%
	%
	\begin{align}
	\delta_F(P[\gamma_{i-1},x],Q[i-1,i]) &\leq \max\{\;\delta_F\big(P(\gamma_{i-1}),Q[i-1,\theta(t)]\big),  \notag  \\
	&\hspace{14mm} \notag \delta_F\big(P[\gamma_{i-1},\gamma^*_i],Q[\theta(t),i]\big), \,\delta_F\big(P[\gamma^*_{i},x],q_i\big)\;\} \leq \eps.  \notag 
	\end{align}

	Since $\gamma_i\geq x$, this implies that all $\gamma_{i}$ exist for all $i=2,\ldots,m$. Note that the procedure $\LEP(P[\gamma_{i-1},n],Q[i-1,i])$ can compute $\gamma_{i}$ by finding the reachable path for $(\gamma_i,i)$ across \FSD$(P[\gamma_{i-1},n],Q[i-1,i])$.
	Therefore \textsc{DecisionAlgorithm$(P,Q,\eps)$} returns ``Yes''.
	%
\end{proof}

\subsubsection{Long Edges with $l_Q>4\eps$} \label{subsubsec:4eps}

Our algorithm also can be applied to the case that one curve has arbitrary edge lengths and the other curve has edge lengths greater than $4\eps$. 

\begin{theorem}[Single Curve with Long Edges]  \label{cor:4eps}
	Let $\eps >0$, and let $P:[1,n]\RD$ and $Q:[1,m]\RD$ be two polygonal curves such that $l_P>0$ and $l_Q>4\eps$. Then there exists a greedy decision algorithm, \algref{directedDecider}, that can determine whether $\delta_F(P,Q) \leq \eps$ or $\delta_F(P,Q)>\eps$ in $O(n+m)$ time.
\end{theorem}

\begin{proof}
	In the proof of~\lemref{pwmatching}, we can replace~\lemref{fp} with~\lemref{3delta-ball}, and realize that \lemref{pwmatching} also holds for the case $l_P>0$ and $l_Q>4\eps$. 
	The rest follows from~\thmref{correctness1} and \thmref{greedyalg}.
\end{proof}

\subsection{Necessity of the Assumption} \label{subsec:assumption}

As we have seen so far, \algref{directedDecider} greedily constructs
a feasible Fr\'echet matching by linearly walking on curve $P$ to find all longest $\eps$-prefixes on it with respect to the corresponding edges of $Q$. Unfortunately, this property is not always true for
curves with short edges. In general,
there can be a quadratic number of blocked regions in the
free space diagram of two curves; see Fig.~\ref{fig:hardness} as
an example of two curves in $\Reals^2$ that have 
edges of length exactly equal $2\eps$ except for some edges with
lengths in $[2\eps,(1+\sqrt{2})\eps]$. This example demonstrates that
our simple greedy construction of a Fr\'echet matching 
is unlikely to work if the edges are shorter than the assumptions we made. It also shows that our greedy construction does not work if both curves have edge lengths of at least $2\eps$.

\begin{figure}[htbp]
	\begin{center}
		\includegraphics[width=12.5 cm]{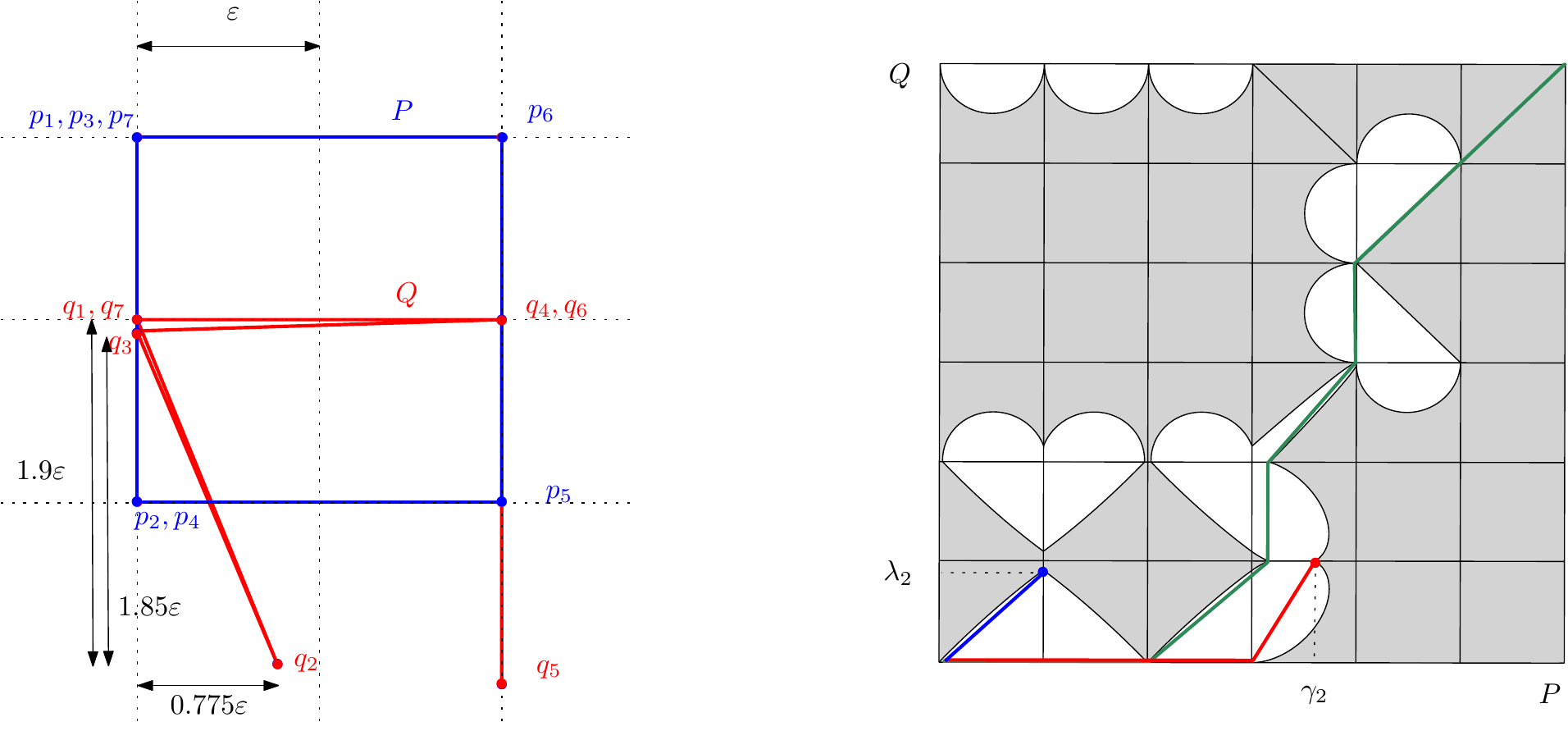}
		\caption{ 
			An example in which the greedy algorithm fails to realize the Fr\'echet matching highlighted in green. Here, $P[1,\gamma_2]$ is the longest $\eps$-prefix in $P$ with respect to $Q[1,2]$, as illustrated by the red reachable path. Also $Q[1,\lambda_2]$ is the longest $\eps$-prefix in $Q$ with respect to $P[1,2]$ as illustrated by the blue reachable path. Every edge is $2\eps$ long, except for the edges $Q[1,2]$ and $Q[2,3]$ that have lengths $2.02\eps$ and $2.005\eps$, respectively. The latter values are still in the range $[2\eps,(1+\sqrt{2})\eps]$.}
		\label{fig:hardness}
	\end{center}
\end{figure}

\section{Optimization and Approximation}\label{sec:optimization}

In this section, we present two algorithms for computing and approximating the~\Frd~between two curves with long edges, respectively. First we give an exact algorithm which runs in  $O((n+m)\log(n+m))$ time. Afterwards, we present a linear time algorithm which is similar to the greedy decision algorithm, but it uses the notion of minimum prefix to approximate the \Frd.

\subsection{Optimization} \label{subsec:opt}

The main idea of our algorithm is that we compute critical values of
the \Frd~between two curves and then perform binary search on these to find the optimal value acquired by the decision
algorithm.
In general, there are a cubic number of \emph{critical values}, which are candidate values for the Fr\'echet distance between two polygonal curves. These critical values are those $\eps$ for which $p_1\in B(q_1,\eps)$ or $p_n\in B(q_m,\eps)$, or when decreasing $\eps$ slightly a free space interval disappears on the boundary of a free space cell or a monotone path in the free space becomes non-monotone. See Alt and Godau~\cite{ag-cfdb-95} for more details on critical values.
In our case we can show that it suffices to consider only a linear number of critical values, because the assumption on the edge lengths of the curves implies that a piecewise orthogonal matching exists, which reduces the number of possible critical values.
%
%
Our optimization algorithm consists of the following four steps:

\begin{enumerate}
	\item Run
	\textsc{DecisionAlgorithm($P,Q,\eps_0$)} with $\eps_0 = \min\{l_P/2,l_Q/(1+\sqrt{d})\}$ and store all $\gamma_i=\LEP(P[\gamma_{i-1},n],Q[i-1,i],\eps_0)$ for all $i=2,\ldots,m$.  Only proceed if \textsc{DecisionAlgorithm($P,Q,\eps_0$)} returns ``Yes''. 
	\item If $P[\gamma_{i-1},\gamma_i]$ is not $(Q[i-1,i],\eps_0)$-monotone for some $i=2,\ldots,m$ then return $\delta_F(P,Q)=\eps_0$. 
	\item Compute $C:=\cup_{i=2}^{m}C_i \cup \{\eps_0 \}$, where $C_i$ is the set of all critical values for $P[\alpha_{i-1},\gamma_{i}]$ and $Q[i-1,i]$. Here, $P(\alpha_i)$ is the first point along $P[\gamma_{i-1},n]$ that intersects $B(q_i,\eps_0)$ and $\alpha_1=1$.
	\item Sort $C$ and perform binary search on $C$ using \textsc{DecisionAlgorithm}($P,Q,\cdot$) to find $\delta_F(P,Q)$.
\end{enumerate}

In step (1) we set $\eps_0=\min\{l_P/2,l_Q/(1+\sqrt{d})\}$. This means
that $l_P\geq 2\eps_0$ and $l_Q\geq (1+\sqrt{d})\eps_0$.
%
Step (2) handles the case that the matching induced by the $\gamma_i$ may not
be a piecewise orthogonal matching.
But once the algorithm proceeds to step (3), there exists a piecewise orthogonal
matching between $P$ and $Q$.  
%
This restricts the set of critical values we have to consider in step (3) as follows:
Let $\eps^*\leq\eps_0$ and assume $\eps^* =\delta_F(P,Q)$. Let $\gamma^*_i=\LEP(P[\gamma^*_{i-1},n],Q[i-1,i],\eps^*)$, for $i=2,\ldots,m$, and let $P(\alpha^*_i)$ be the first intersection point between $P[\gamma^*_{i-1},n]$ and $B(q_i,\eps^*)$, and $\alpha^*_1=\gamma^*_1=1$.
From $B(q_i,\eps^*)\subseteq B(q_i,\eps_0)$ follows that $\alpha_i \leq \alpha^*_i \leq \gamma^*_i \leq \gamma_{i}$. And since $\gamma_{i-1}\leq\gamma_i$, we know that
$P[\gamma^*_{i-1},\gamma^*_{i}] \subseteq P[\alpha_{i-1},\gamma_i]$. We thus have observed the following, see Fig.~\ref{fig:shrunk}: 
%
\begin{observation} \label{obs:superpath}
	Let $\eps^*\leq \eps_0$. For all $i=2,\ldots,m$:
	
	\begin{enumerate*}
		\item $\alpha_i \leq \alpha^*_i \leq \gamma^*_i \leq \gamma_{i}$,
		\item $P[\gamma^*_{i-1},\gamma^*_{i}] \subseteq P[\alpha_{i-1},\gamma_{i}]$.
	\end{enumerate*}
	
\end{observation}
Therefore all critical values for $P[\gamma^*_{i-1},\gamma^*_{i}]$ and
$Q[i-1,i]$ must be contained in the set $C_i$ which are the critical
values for $P[\alpha_{i-1},\gamma_{i}]$ and $Q[i-1,i]$, and the binary search in step (4) will
identify $\eps^*$.

\begin{lemma} [Correctness] \label{lem:optimizationCorrectness}
	Let $\eps_0=\min\{l_P/2,l_Q/(1+\sqrt{d})\}$ and let
	$\eps^*=\delta_F(P,Q)$. If in step (1) of the optimization algorithm
	\textsc{DecisionAlgorithm$(P,Q,\eps_0)$} returns ``Yes'', then the
	optimization algorithm returns $\eps^*$ and
	$\eps^*\leq\eps_0$. Otherwise $\eps^*\geq\eps_0$.
\end{lemma}

\begin{proof}
	If \textsc{DecisionAlgorithm$(P,Q,\eps_0)$} returns ``No'' then
	\thmref{equal2eps} implies that $\delta_F(P,Q)=\eps^*\geq\eps_0$.
	Now suppose, for the remainder of this proof, that
	\textsc{DecisionAlgorithm$(P,Q,\eps_0)$} returns ``Yes''. Then we know
	that all $\gamma_i$ exist and $\delta_F(P[\gamma_{i-1},\gamma_i],Q[i-1,i])\leq\eps_0$ for all
	$i=2,\ldots,m$, and therefore $\eps^*\leq\eps_0$, see also \thmref{equal2eps}.
	This implies that $\eps_0$ is an upper bound on all critical values in $C$. 
	It remains to show that the optimization algorithm returns
	$\eps^*$.

	If in step (2) there is an $i=2,\ldots,m$ such that
	$P[\gamma_{i-1},\gamma_{i}]$ is not $(Q[i-1,i],\eps)$-monotone,
	then there must exist an edge $P[j,j+1]$, for $\gamma_{i-1}\leq j<\gamma_i$,
	such that the angle between the direction vectors of $P[j,j+1]$ and $Q[i-1,i]$
	is greater than $\pi/2$.
	The length of all edges in $P$ must be at least $2\eps_0$. 
	But for this edge, the only way a (monotone) Fr\'echet
	matching between $P[\gamma_{i-1},\gamma_i]$ and $Q[i-1,i]$ of width at
	most $\eps_0$ can exist is if $\|p_{j+1}-p_j\|= 2\eps_0$ and both
	$p_j$ and $p_{j+1}$ are matched to $x=B(p_j,\eps_0)\cap
	B(p_{j+1},\eps_0)\cap Q[i-1,i]$. Therefore the width of such a Fr\'echet
	matching is exactly $\eps_0$ and $\eps^*=\eps_0$.
	%
	
	It remains to show that if the algorithm passes step (2)
	it returns $\eps^*$ at the end of step (4). Since $\eps_0\in C$ and $\eps^*\leq\eps_0$, the binary search will return $\eps^*$ if $\eps^*=\eps_0$.  
	%
	So assume now that $\eps^*<\eps_0$.
	Since the algorithm passes step (2), it follows from \lemref{monotonicity} that the matching induced by the
	$\gamma_i$ is indeed a piecewise orthogonal matching of width less
	than $\eps_0$.
	From \obsref{superpath} follows that
	all critical values for $P[\gamma^*_{i-1},\gamma^*_{i}]$ and
	$Q[i-1,i]$ must be contained in the set $C_i$ of all critical
	values for $P[\alpha_{i-1},\gamma_{i}]$ and $Q[i-1,i]$. Thus,
	$\eps^*\in C=\cup_{i=2}^m C_i$, and the binary search in step (4)
	returns $\eps^*$.

\end{proof}
%

%
%

\noindent\textbf{Computing The Critical Values:}
A piecewise orthogonal matching of width $\eps^*$ between $P$ and $Q$
is comprised of orthogonal matchings between
$P[\gamma^*_{i-1},\gamma^*_{i}]$ and $Q[i-1,i]$ for all
$i=2,\ldots,m$. The piecewise orthogonal matching may map vertices from $P$ to
$Q[i-1,i]$ either by an orthogonal projection or by mapping to the
endpoints $q_{i-1}, q_i$.
And vertices $q_i$ may be mapped by on orthogonal projection to $P[\alpha_i,\gamma_{i}]$.
These mappings define point-to-point distances that are candidates for $\eps^*$, and thus
critical values between $P[\gamma^*_{i-1},\gamma^*_{i}]$ and $Q[i-1,i]$ that we need to optimize over.
But since $\eps^*$ is not known beforehand, we compute the superset $C_i$ of 
critical values between
$P[\alpha_{i-1},\gamma_{i}]$ and $Q[i-1,i]$ as follows:
%
%
%
Let $\H_1$ be the hyperplane perpendicular to $Q[i-1,i]$ and tangent to $B(q_{i-1},\eps_0)$ that intersects $Q[i-1,i]$. Similarly, define $\H_2$ with respect to $B(q_i,\eps_0)$. 
For each $p_j\in P[\alpha_{i-1},\gamma_{i}]$:
(1) If $p_j$ lies between $\H_1$ and $\H_2$, then any orthogonal
matching of width $\eps^*$ maps $p_j$ to its orthogonal projection on
$Q[i-1,i]$. We therefore add the distance $\|p_j,Q[i-1,i]\|$ to
$C_i$.
(2) If $p_j$ lies on the $q_{i-1}$-side of $\H_{1}$, then an
orthogonal matching of width $\eps^*$ can map $p_j$ either to
$q_{i-1}$ or to its orthogonal projection on $Q[i-1,i]$. In this case
we store both $\|p_j,Q[i-1,i]\|$ and $\|p_j-q_{i-1}\|$ in
$C_i$.
Similarly, if $p_j$ lies on the $q_i$-side of $\H_2$ then we store 
$\|p_j,Q[i-1,i]\|$ and $\|p_j-q_{i}\|$ in $C_i$. 
Finally, for each edge $e$ in $P[\alpha_i,\gamma_i]$: (3) we store $\|q_i,e\|$. See Fig.~\ref{fig:shrunk} for more illustration. 
%
%
%
%
We have the following theorem:
\begin{figure}[htbp]
	\centering 
	\includegraphics[width=12.5 cm]{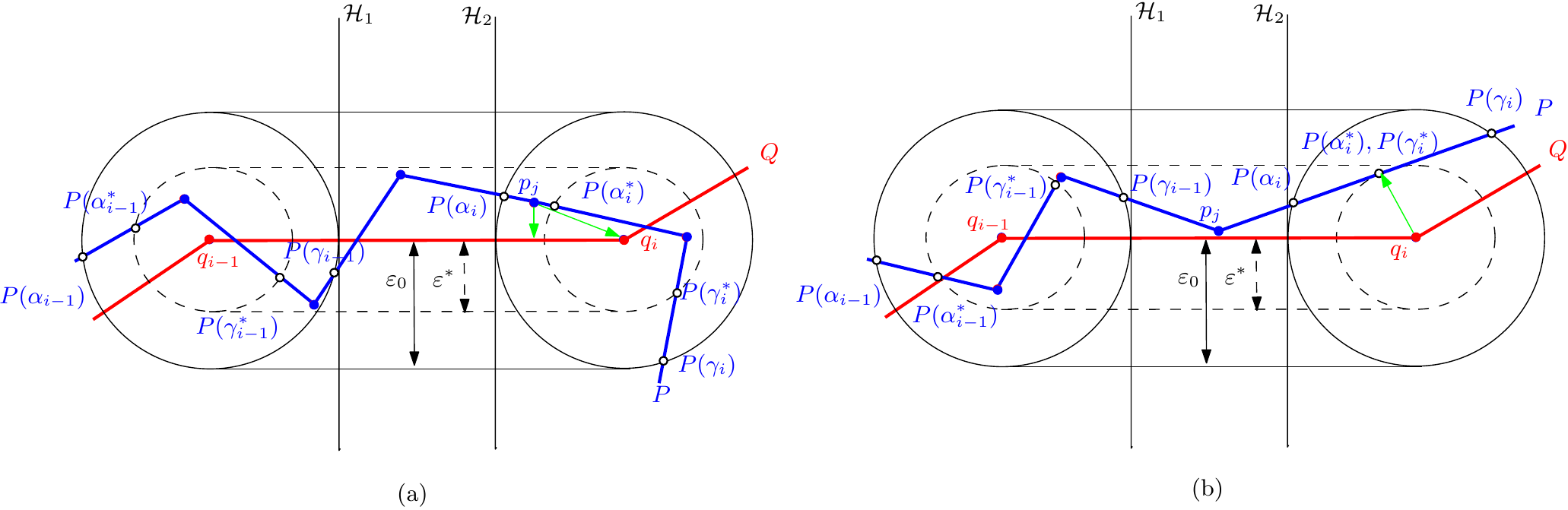}
	\caption{
		Shown are two examples of orthogonal matchings  between
		$P[\gamma^*_{i-1},\gamma^*_{i}]$ and $Q[i-1,i]$, and
		the associated critical values (point-to-point distances
		defined by the matching). 
		The cylinders $C(Q[i-1,i],\eps_0)$
		and $C(Q[i-1,i],\eps^*)$ are shown, where $\eps_0\geq \eps^*$.
		(a) $p_j$ falls into case (2), when the orthogonal matching maps $p_j$
		either to $q_i$ (if $p_j$ lies inside $B(q_i,\eps^*)$) or orthogonally
		to $Q[i-1,i]$ (if $p_j$ lies outside $B(q_i,\eps^*)$).
		(b) If an edge of $P$ is tangent to $B(q_i,\eps^*)$, then case (3) occurs.
		Here, the orthogonal matching has to
		map $q_i$ to an edge $e$ in $P[\alpha_i,\gamma_{i}]$. }
	\label{fig:shrunk}
\end{figure}

\begin{theorem}[Optimization] \label{thm:opt}
	Let $P:[1,n]\RD$ and $Q:[1,m]\RD$ be two polygonal curves. If $\delta_F(P,Q) < \min\{l_P/2, l_Q/(1+\sqrt{d})\}$, then $\delta_F(P,Q)$ can be computed in $O((n+m)\log (n+m))$ time.
\end{theorem}

\begin{proof}
	By \lemref{optimizationCorrectness} we know that the optimization
	algorithm returns $\delta_F(P,Q)$ correctly if $\delta_F(P,Q)$ is strictly less than $\min\{l_P/2, l_Q/(1+\sqrt{d})\}$. 
	It only remains to prove the runtime of the algorithm.
	First we show that the number of critical values is linear. For each
	segment $Q[i-1,i]$, there are three cases for critical values contained in $C_i$:
	(1) There are  at most
	$\ceil{\alpha_{i}-\gamma_{i-1}}+1$ values if vertex $p_j$ lies between $\H_1$
	and $\H_2$. This is an upper
	bound for the number of vertices in $P[\gamma_{i-1},\alpha_{i}]$.
	(2) There are at most $2(\ceil{\gamma_{i-1}-\alpha_{i-1}}+1)$ values if
	vertex $p_j$ lies on the $q_{i-1}$-side of $\H_1$, and similarly
	there are  at most $2(\ceil{\gamma_{i}-\alpha_{i}}+1)$ values 
	if $p_j$ lies on the $q_i$-side of $\H_2$.
	(3) There are at most $\ceil{\gamma_{i}-\alpha_{i}}$ critical values for each
	edge $e$ in $P[\alpha_{i},\gamma_{i}]$.
	Overall, the total is: 
	$|C_i| \leq (\gamma_{i}-\alpha_{i}) +
	(\gamma_{i-1}-\alpha_{i-1}) + (\gamma_{i} - \alpha_{i-1}) + 11 <
	2(\gamma_{i}-\alpha_{i-1}) + 11 =
	2(\gamma_{i}-\gamma_{i-1})+2(\gamma_{i-1}-\alpha_{i-1})+11$.
	The latter inequality follows because $\alpha_{i-1}\leq
	\gamma_{i-1}\leq \alpha_i\leq \gamma_i$, see Fig.~\ref{fig:shrunk}.
	Note that $\sum_{i=2}^{m}(\gamma_{i}-\gamma_{i-1})= \gamma_m-\gamma_1=
	n-1$, and $\sum_{i=2}^{m}(\gamma_{i-1}-\alpha_{i-1}) < n$, therefore,
	$|C|=\sum_{i=2}^{m}|C_i| < 2(n-1)+ 2n + 11(m-1) = O(n+m)$.
	
	The optimization algorithm first runs \algref{directedDecider} in $O(n+m)$ time, then computes $C$ in
	$O(n+m)$ time, and finally sorts $C$ in $O((n+m)\log (n+m))$ time and
	performs binary search on $C$ using the decision algorithm in
	$O((n+m)\log (n+m))$ time. Therefore, the total runtime is
	$O((n+m)\log (n+m))$.
\end{proof}

\subsection{Approximation Algorithm} \label{subsec:approximation}
In this section we present a $\sqrt{d}$-approximation algorithm running in linear time. 
As a counterpart to the notion of longest $\eps$-prefix we now introduce the notion of minimum prefix, which is the longest prefix of $P$ with minimum Fr\'echet distance to a line segment $e$.

\begin{definition}[Minimum Prefix] \label{def:prefix}
	Let $P:[1,n]\RD$ be a polygonal curve and $e:[1,2]\RD$ be a segment. Define $\gamma'=\max \argmin_{1 \leq t \leq n}{\delta_F(P[1,t],e)}$. We call $P[1,\gamma']$ the {\em minimum prefix of $P$ with respect to $e$}.
\end{definition}

Note that in the definition above, $P(\gamma')$ necessarily lies on the boundary of $B(e_2,\eps')$, where $\eps'=\min_{1 \leq t \leq n}{\delta_F(P[1,t],e)}$.
The approximation algorithm is presented in \algref{optalgorithm}.
First, for an initial threshold $\eps_0=\min\{l_P/(2\sqrt{d}), l_Q/(2d)\}$, it runs the decision algorithm, i.e., \textsc{DecisionAlgorithm($P,Q,\eps_0$)}. The algorithm only continues if ``Yes'' gets returned. 
This ensures that $P$ and $Q$ have long edges, with $l_P\geq 2\sqrt{d}\eps>2\eps$ and 
$l_Q\geq 2d\eps>(1+\sqrt{d})\eps$.
Then, similar to the decision algorithm, the approximation
algorithm greedily searches for longest $\eps$-prefixes with respect to each
segment of $Q$. However, it updates the current value of $\eps$ in
each step, by computing the minimum prefix and its associated \Frd~to
the portion of $Q$ considered so far. 

\begin{algorithm}
	\DontPrintSemicolon
	\SetKwFunction{ApproximationAlgorithm}{\textsc{ApproximationAlgorithm}}
	\SetKwFunction{Prefix}{\textsc{Prefix}}
	\SetKwFunction{LongestDeltaPrefix}{\LDP}
	\SetKwFunction{MinimumPrefix}{\MinPrefix}
	\SetKwProg{myalg}{}{}{}
	\caption{Approximate $\delta_F(P,Q)$}
	\label{alg:optalgorithm}
	
	\BlankLine
	\myalg{\ApproximationAlgorithm{$P[1, n], Q[1, m]$}}{ 
		$\eps_0 \leftarrow \min \{l_P/2\sqrt{d},l_Q/2{d}\}$\;
		\lIf{\DecisionAlgorithm{$P,Q,\eps_0$}= ``No''}
		{
			\KwRet{``I don't know''} \label{algline:dontknow}
		}
		{
			$(\gamma_2, \eps_2)\leftarrow$\MinimumPrefix{$P[1,n], Q[1,2]$}\; \label{algline:minprefix-2}
			$\eps \leftarrow \eps_2$\;
			$s\leftarrow \gamma_2$\;
			\For{$i \leftarrow 3$ \KwTo $m$}
			{
				$(\gamma_i, \eps_i)\leftarrow$\MinimumPrefix{$P[s,n], Q[i-1,i]$}\; 
				$\eps \leftarrow \max \lbrace \eps, \eps_i \rbrace$ \;
				$s\leftarrow \gamma_i$\;
			}	
			\eIf{$\gamma_m=n$}{\KwRet{$\eps$}}	
			{
				$\eps \leftarrow \max \lbrace \eps, \delta_F(P[\gamma_m,n],q_m) \rbrace $\;
				\KwRet{$\eps$}
			}
		}
	}
\end{algorithm}

Now we are ready to prove the correctness of \algref{optalgorithm}:
\begin{lemma}[The Approximation] \label{lem:optcorrect}
	Let $P=P[1,n]$ and $Q=Q[1,m]$ be two polygonal curves and let $\eps^*=\delta_F(P,Q)$. If $\eps^* \leq \min\{l_P/2\sqrt{d}, l_Q/2d\}$ then \textsc{ApproximationAlgorithm}($P, Q$) returns a value between $\eps^*$ and $\sqrt{d}\eps^*$. Otherwise it returns ``I don't know''.
\end{lemma}

\begin{proof} 
	From \algref{optalgorithm} we have that $\eps_i=\delta_F(P[\gamma_{i-1},\gamma_i], Q[i-1,i])$. We prove by induction on $i$ that $\eps_i\leq\sqrt{d}\eps^*$.
	For $i=2$, $\eps_2$ is being minimized and obviously $\eps_2 \leq \eps^* < \sqrt{d}\eps^*$. 
	For any $i>2$, there are two possible cases: either $\eps_i \leq \eps^*$ or $\eps_i > \eps^*$. In the former case, trivially 
	$\eps_i< \sqrt{d}\eps^*$. In the remainder of the proof we consider the latter case that is
	$\eps_i > \eps^*$. We know from \thmref{correctness1} that all $\gamma^*_i=\LEP(P[\gamma^*_{i-1},n],Q[i-1,i],\eps^*)$ for all $i=1,2,\ldots, m$ exist. And by inductive hypothesis we know that $\max  \lbrace \eps_2,\ldots , \eps_{i-1} \rbrace \leq \sqrt{d}\eps^*$. 
	
	We also know from line 8 of \algref{optalgorithm} that $\eps_i=\delta_F(P[\gamma_{i-1},\gamma_i], Q[i-1,i])$ and $P[\gamma_{i-1},\gamma_{i}]$ is the minimum prefix with respect to $Q[i-1,i]$. For the sake of contradiction we assume $\eps_i>\sqrt{d}\eps^*>\eps^*$. 
	We now distinguish two cases:
	
	(a) If $\gamma_{i-1}<\gamma^*_{i-1}$, then by \lemref{fp} we have $\delta_F(P[\gamma_{i-1}, \gamma^*_{i-1}],q_{i-1}) \leq \sqrt{d}\eps^*$. Also $\delta_F(P[\gamma^*_{i-1},\gamma^*_i],Q[i-1,i])\leq \eps^*$, hence $\delta_F(P[\gamma_{i-1},\gamma^*_i],Q[i-1,i])\leq \sqrt{d}\eps^* < \eps_i$. This contradicts the fact that $P[\gamma_{i-1},\gamma_{i}]$ is the minimum prefix of $P[\gamma_{i-1},n]$ with respect to $Q[i-1,i]$, see Fig.~\ref{fig:optproof}(a).
	%
	
	(b) Now for the case that $\gamma^*_{i-1} < \gamma_{i-1}$, consider the matching $(\sigma,\theta)$ realizing $\delta_F(P,Q)=\eps^*$. There exists some $t\in [0,1]$ such that $\gamma_{i-1}=\sigma(t)$. We can see that $Q(\theta(t))\in Q[i-1,i]$ as follows: We know that $B(q_{i-1},\eps_{i-1})\cap B(q_i,\eps^*)=\emptyset$ since $\|q_{i-1}-q_i\|\geq 2d\eps^*>(1+\sqrt{d})\eps^*$ and $\eps_{i-1}\leq \sqrt{d}\eps^*$. This implies $\gamma_{i-1} < \gamma^*_i$ and therefore $\gamma^*_{i-1} < \gamma_{i-1} < \gamma^*_i$, and correspondingly ${i-1}\leq \theta(t)\leq i$. By inductive hypothesis we know that $\eps_{i-1}=\|q_{i-1}-P(\gamma_{i-1})\|\leq \sqrt{d}\eps^*$, thus $Q[i-1,\theta(t)] \subseteq B(P(\gamma_{i-1}),\sqrt{d}\eps^*)$ which implies $\delta_F(P(\gamma_{i-1}),Q[i-1,\theta(t)])\leq \sqrt{d}\eps^*$. 
	Combining this with $\delta_F(P[\gamma_{i-1},\gamma^*_i],Q[\theta(t),i])\leq \eps^*$ from the optimal matching yields $\delta_F(P[\gamma_{i-1},\gamma^*_i],Q[i-1,i])\leq \sqrt{d}\eps^* < \eps_i$. This contradicts that $P[\gamma_{i-1},\gamma_{i}]$ is the minimum prefix of $P[\gamma_{i-1},n]$ with respect to $Q[i-1,i]$, see Fig.~\ref{fig:optproof}(b). 
	
	In the end, if $\gamma_m<n=\gamma^*_m$, then \lemref{fp} again implies $\delta_F(P[\gamma_m,n],q_m ) \leq \sqrt{d}\eps^*$ as claimed. 
	The algorithm returns $\max\{\eps_2,\ldots,\eps_m\}$. Since there has to be some $\eps_j>\eps^*$, and we proved by induction that all $\eps_i\leq \sqrt{d}\eps^*$, the algorithm returns a value between $\eps^*$ and $\sqrt{d}\eps^*$.
\end{proof}

\begin{figure}[htbp]
	\begin{center}
		\includegraphics[width=12.5 cm]{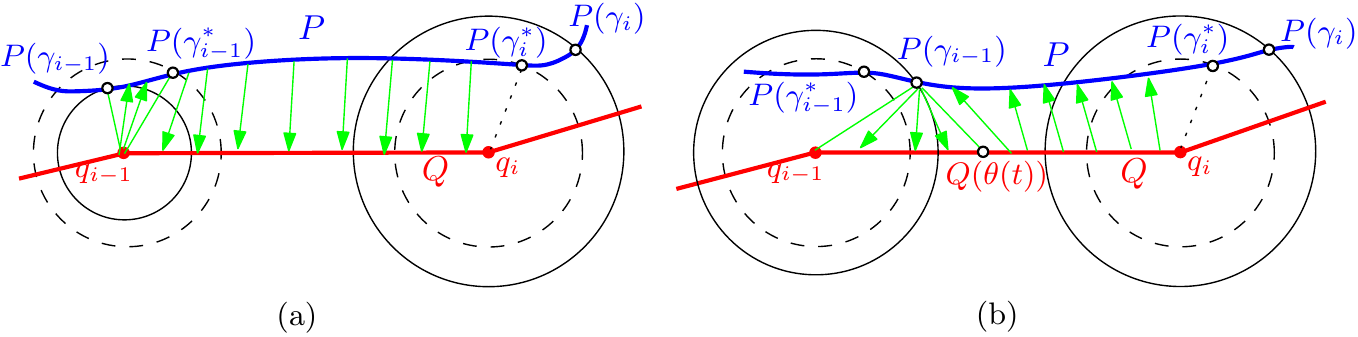}
		\caption{Illustration for the proof of~\lemref{optcorrect} when $\eps_i> \eps^*$. (a) $\gamma_{i-1}<\gamma^*_{i-1}$ (b) $\gamma_{i-1} > \gamma^*_{i-1}$.}
		\label{fig:optproof}
	\end{center}
\end{figure}

\noindent\textbf{The \MinPrefix~Procedure:}
Given a polygonal curve $P:[1,n]\RD$ and a segment $e:[1,2]\RD$, we implement \MinPrefix($P,e$), as described in \algref{minPrefix}, as follows: 
For every $i=1,\ldots,n-1$, let $c_i$ be the distance associated with a minimum prefix ending on the segment $P[i,i+1]$. Formally, $c_i=\min_{t\in[i,i+1]}\delta_F(P[1,t],e)$. 
\algref{minPrefix} computes all the $c_i$ in a dynamic programming fashion. The minimum of the $c_i$ is the desired $\eps$, and the \LEP\ computes the corresponding $\gamma$.
%

\begin{algorithm}
	\DontPrintSemicolon
	\SetKwFunction{MinimumPrefix}{\MinPrefix}
	\SetKwProg{myalg}{}{}{}
	\SetKwComment{Comment}{}{}
	\caption{Compute \MinPrefix($P[1,n],e[1,2]$)} 
	\label{alg:minPrefix}
	
	\BlankLine
	\myalg{\MinPrefix{$(P[1, n], e[1, 2])$}}{ 
		$c \leftarrow ||p_1-e_1||$\; 
		$\eps' \leftarrow \min\{l_P/2,\|e\|/2\sqrt{d}\}$\;
		$\gamma' \leftarrow \LEP(P[1,n],e[1,2],\eps')$\; \label{algline:minprefix-callLDP}
		
		\For{$i \leftarrow 1$ \KwTo $\lfloor\gamma'\rfloor$}
		{
			$c_i \leftarrow \max\{c,||e_2,P[i,i+1]||\}$\;
			$c \leftarrow \max\{c, ||p_{i+1},e[1,2]||\}$
		}
		$\eps=\min_{1\leq i\leq \lfloor\gamma'\rfloor} c_i$\;
		\KwRet{$(\eps$, \LEP$(P[1,n], e[1,2], \eps))$}
	}
\end{algorithm}

Before we can prove the correctness of~\algref{minPrefix}, we need the following technical lemma that states when $\eps$ is increased, the longest $\eps$-prefix can only get longer.
\begin{lemma}[Prefix monotonicity] 
	\label{lem:prefixmon}
	Let $P:[1,n]\RD$ and $Q:[1,m]\RD$ be two polygonal curves and $\eps > \eps'>0$. Let $\gamma_1=\gamma'_1=1$, $\gamma_{i}= \LEP(P[\gamma_{i-1},n],Q[i-1,i],\eps)$ and $\gamma_i'=\LEP(P[\gamma'_{i-1},n],Q[i-1,i],\eps')$  for all $i=2,\ldots, m$.
	Then $\gamma_i < \gamma_i'$ for all $i=2,\ldots, m$.
\end{lemma} 

\begin{proof}
	The proof is by the induction. For $i=2$, we know that $\delta_F(P[1,\gamma'_2],Q[1,2])\leq \eps' < \eps$. Let $x$ be a parameter such that $P(x)$ is the first intersection point between $P[\gamma'_2,n]$ and the boundary of $B(q_2,\eps)$, thus $\gamma'_2 < x$. Now observe that $\delta_F(P[\gamma'_2, x],q_2)\leq \eps$. Combining $\delta_F(P[1,\gamma'_2],Q[1,2]) < \eps$ and $\delta_F(P[\gamma'_2, x],q_2)\leq \eps$ yields $\delta_F(P[1,x],Q[1,2])\leq \eps$. Since $\gamma_2$ is the longest $\eps$-prefix with respect to $Q[1,2]$, we have $\gamma'_2 <  x \leq \gamma_2$, and  therefore $\gamma'_2 < \gamma_2$. Now for $i>2$, by the inductive hypothesis we have $\gamma'_{i-1} < \gamma_{i-1}$. It remains to show $\gamma'_{i}<\gamma_{i}$. Consider a matching $(\sigma,\theta)$ realizing $\delta_F(P[\gamma'_{i-1},\gamma'_{i}],Q[i-1,i])\leq \eps'$. Let $t$ be the value such that $\sigma(t)=\gamma_{i-1}$. Now we construct a new matching for $P[\gamma_{i-1},x]$, where $x$ is defined as in the inductive base, but with respect to $B(q_i,\eps)$. We know that $\delta_F(P(\gamma_{i-1}),Q[i-1,\theta(t)])\leq \eps$. Also we have $\delta_F(P[\gamma_{i-1},\gamma'_{i}],Q[\theta(t),i])\leq \eps' < \eps$ by $(\sigma,\theta)$. Observe that $\delta_F(P[\gamma'_{i},x],q_i)\leq \eps$. Thus, $\delta_F(P[\gamma_{i-1},x],Q[i-1,i])\leq \eps$ and using a similar argument as in the inductive base we have $\gamma'_{i} < x \leq \gamma_{i}$, therefore $\gamma'_{i}<\gamma_{i}$.
\end{proof}

Now we are ready to prove the correctness of~\algref{minPrefix}:

\begin{lemma}[Correctness] \label{lem:minprefixcrrct}
	Let $e:[1,2]\RD$ be a line segment and let $P:[1,n']\RD$ be a polygonal curve monotone with respect to the line supporting $e$. The distance returned by \MinPrefix$(P,e)$ is $\min_{1 \leq t \leq n'}{\delta_F(P[1,t],e)}$.
\end{lemma}
\begin{proof}
	According to the algorithm: $$c_i=\max\{\|p_1-e_1\|, \max_{1 \leq j
		\leq i-1} \|p_{j+1},e\|, \|e_2, P[i,i+1]\|\}.$$ Since $e[1,2]$ is a
	segment and $P[1,n']$ is monotone with respect to the line supporting
	$e$, it follows from \lemref{monotonicity} that for any $i\leq t\leq i+1$ there exists an orthogonal matching such that: 
	%
	
	$$\delta_F(P[1,t],e)=\max\{\|p_1-e_1\|,
	\max_{1 \leq j \leq i-1}\|p_{j+1},e\|, \|P(t)-e_2\|\}$$
	By taking the minimum on both sides, we get: 
	$$\min_{i\leq t \leq i+1}\delta_F(P[1,t],e) = \max\{\|p_1-e_1\|,\max_{1 \leq j \leq i-1}\|p_{j+1},e\|, \min_{i\leq t \leq i+1}\|P(t)-e_2\|\} = c_i.$$
	It suffices to run the for-loop until $n'=\lfloor\gamma'\rfloor$, since by the assumption we only compute the 
	minimum $\eps$-prefix $P[1,\gamma]$ if its distance is at most $\eps'$ (line 3 of \algref{minPrefix}), and by \lemref{prefixmon}
	it follows $\gamma<\gamma'$. Therefore, $\eps=\min^{\lfloor\gamma'\rfloor}_{i=1}{c_i}=\min^{ n'}_{i=1}{c_i}=\min_{1 \leq t \leq n'}{\delta_F(P[1,t],e)}.$
\end{proof}


\begin{theorem}[Runtime]
	Let $P:[1,n]\RD$ and $Q:[1,m]\RD$ be two polygonal curves. If $\delta_F(P,Q) \leq \min\{l_P/(2\sqrt{d}), l_Q/(2d)\}$, then \algref{optalgorithm} approximates $\delta_F(P,Q)$ in $O(n+m)$ time within an approximation factor of $\sqrt{d}$. 
\end{theorem}

\begin{proof}
	Let $\eps^*=\delta_F(P,Q)$. The algorithm only proceeds past \alglineref{dontknow} if 
	$\eps^*\leq \eps_0=\min\{l_P/2\sqrt{d}, l_Q/2d\}$ and $\textsc{DecisionAlgorithm}(P,Q,\eps_0)$ returns ``Yes''. 
	Now, let $\eps'=\sqrt{d}\eps_0$, $\gamma'_1=1$,  and for all
	$i=2,\ldots, m$ let
	$\gamma'_i=\LEP(P[\gamma'_{i-1},n],Q[i-1,i],\eps')$.
	Note that by definition of $\eps'$, both curves have long edges, i.e., 
	$l_P\geq2\sqrt{d}\eps_0 > 2\eps'$ and $l_Q\geq 2d\eps_0 = 2\sqrt{d}\eps' >(1+\sqrt{d})\eps'$. 
	%
	%
	%
	From the proof of \lemref{optcorrect} we know that $\eps_i\leq \sqrt{d}\eps^*\leq \sqrt{d}\eps_0 = \eps'$ and since $\|q_{i-1}- q_i \|> 2\sqrt{d}\eps'$, we have that $B(q_{i-1},\eps')\cap B(q_i,\eps_i)=\emptyset$. Therefore, $\gamma'_{i-1}<\gamma_{i}$. \lemref{prefixmon} implies that $\gamma_i\leq \gamma'_i$ due to $\eps_i \leq \eps'$, therefore $\gamma_{i-1}< \gamma'_{i-1} <\gamma_{i} <\gamma'_{i}$ for all $i= 2 ,\ldots,m $. 
	
	The for-loop in \algref{optalgorithm} has $m-2$ iterations. In iteration $i$, the algorithm calls $\MinPrefix(P[\gamma_{i-1},\gamma'_{i}],Q[i-1,i])$ in line 8. The for-loop in \algref{minPrefix} has $\ceil{\gamma'_{i}- \gamma_{i-1}}+1$ iterations,  where  $\ceil{\gamma'_{i}-\gamma_{i-1}}+1$ is the upper bound for the number of vertices in $P[\gamma_{i-1},\gamma'_i]$. Therefore, the runtime of \algref{optalgorithm} is: $\sum_{i=2}^{m} (\ceil{\gamma'_{i}-\gamma_{i-1}}+1) \leq \sum_{i=2}^{m} ({\gamma'_{i}-\gamma_{i-1}}+2) = \sum_{i=2}^{m} (\gamma'_{i}-\gamma_{i})+ \sum_{i=2}^{m} ({\gamma_{i}-\gamma_{i-1}})+ 2(m-1)$. Since $\gamma'_{i-1}\leq \gamma_{i}$, we have $\sum_{i=2}^{m} (\gamma'_{i}-\gamma_{i}) \leq \sum_{i=2}^{m} (\gamma'_{i}-\gamma'_{i-1})$. Thus, $\sum_{i=2}^{m} (\gamma'_{i}-\gamma_{i})+ \sum_{i=2}^{m} ({\gamma_{i}-\gamma_{i-1}})+ 2(m-1) \leq \sum_{i=2}^{m} (\gamma'_{i}-\gamma'_{i-1})+ \sum_{i=2}^{m} ({\gamma_{i}-\gamma_{i-1}})+ 2(m-1)= \gamma'_m-\gamma'_1 + \gamma_m-\gamma_1 + 2(m-1) = 2(n-1)+2(n-1)+2(m-1)= O(n+m)$.
\end{proof}

\section{Data Structure For Longest $\eps$-Prefix Queries}\label{sec:datastructure}

In this section, we consider query variants of the setting in \secref{greedyDecision} for curves in the plane. We wish to solve the following problem: Preprocess a polygonal curve $P:[1,n]\rightarrow \Reals^2$ into a data structure such that for any polygonal query curve $Q:[1,m]\rightarrow \Reals^2$ and a positive $\eps<\min\{l_P/2,l_Q/(1+\sqrt{2})\}$ one can efficiently decide whether $\delta_F(P,Q)\leq \eps$. Note that throughout this section we assume, as before, that $P$ and $Q$ have long edges, i.e., $l_P>2\eps$ and $l_Q>(1+\sqrt{2})\eps$.
Our query algorithm is identical to \algref{directedDecider}. However, the key idea for speeding up the query algorithm is to efficiently compute $\LEP(P[1,n],Q[1,2],\eps)$ for a given query segment $Q[1,2]$ in sublinear time. Our algorithm to compute the longest $\eps$-prefix with respect to $Q[1,2]$ is shown in \algref{queryalg}. 
According to \lemref{monotonicity} if $\delta(P[1,\gamma],Q[1,2])\leq \eps$, then $P[1,\gamma]$ is $(Q[1,2],\eps)$-monotone. This is equivalent to computing the largest parameter $1<t\leq n$ such that the following conditions hold: (1) $p_1\in B(q_1,\eps)$ and $P(t)\in B(q_2,\eps)$, (2)  $P[1,t]\subseteq C(Q[1,2],\eps)$, and (3) $P[1,t]$ is monotone with respect to line supporting $Q[1,2]$. Note that the smallest value $t$ that violates either of the conditions above is a potential $\gamma$.

\begin{algorithm}[htbp]
	\DontPrintSemicolon
	\SetKwFunction{LongetEpsilonPrefix}{\LEP}
	\SetKwFunction{LongestMonotoneSubcurve}{\LMS}
	\SetKwFunction{FirstIntersection}{\FIP}
	\SetKwFunction{LastIntersection}{\LIP}
	\SetKwFunction{CylinderIntersection}{\CI}
	
	\SetKwProg{myalg}{}{}{}
	\SetKwComment{Comment}{}{}
	\caption{Compute \LEP($P[1,n],Q[1,2],\eps$)}
	\label{alg:queryalg}
	\BlankLine
	\myalg{\LEP{$(P[1, n], Q[1, 2],\eps)$}}{ 
		
		\lIf{$p_1\notin B(q_1,\eps)$}{\KwRet{`null'}}
		$\lambda\leftarrow$ \textsc{LongestMonotonePrefix}($P[1,n],Q[1,2]$)\;
		
		$\alpha\leftarrow$ \textsc{FirstIntersection}($P[1,\lambda],B(q_2,\eps))$\;
		
		\lIf{$\alpha=null$}{\KwRet{`null'}}
		
		$\beta\leftarrow$ \textsc{LastIntersection}($P[1,\lambda],B(q_2,\eps))$\;
		
		$r\leftarrow$ \textsc{CylinderIntersection($P[1,\lambda],C(Q[1,2],\eps)$)}\;
		
		\lIf{$r=null$}{\KwRet{$\min(\lambda,\beta)$}}
		
		\lIf {$r<\alpha$ {\bf or} $\lambda < \alpha$}{\KwRet{`null'}}
		
		\lIf {$\alpha< r < \beta$ {\bf or} $\alpha< \lambda < \beta$} {\KwRet $\min(r,\lambda)$}
		
		\lIf {$r> \beta$ {\bf and} $\lambda > \beta$}{\KwRet $\beta$}
	}
\end{algorithm}

%






Here, \textsc{LongestMonotonePrefix} returns $\lambda$, where $P(\lambda)$ is the endpoint of the longest subcurve of $P[1,n]$ that starts in $p_1$ and is monotone with respect to the line supporting $Q[1,2]$. \textsc{FirstIntersection} returns $\alpha$, where $P(\alpha)$ is the first intersection point between $P[1,\lambda]$ and $B(q_2,\eps)$. Similarly, \textsc{LastIntersection} returns $\beta$, where $P(\beta)$ is the last intersection point. \textsc{CylinderIntersection} finds $r$ where $P(r)$ is the first point along $P$ that intersects the boundary of $C(Q[1,2],\eps)$.
%


\noindent{\bf Computing} \textsc{LongestMonotonePrefix}{\bf:} We store
all the edges of $P$ in the leaves of a binary tree $T$ ordered
with respect to their indices. We call the subset of edges stored in the leaves of the subtree rooted at a node $v$ the \emph{canonical subset} of $v$. A set of nodes $v_1,\cdots,v_k$ in the subtree of $v$ is called a set of \emph{canonical nodes} of $v$ if their leaves sets are disjoint and the union of their leaves sets is the leaves of the subtree of $v$.
For each edge in $P$ we consider its direction vector. Each internal node $v$ stores the pair of the minimum/maximum angles between the direction vector and $x$-axis among all associated direction vectors stored in its canonical subset. Once given a query angle $\Phi$ and a starting point $p_1$, we retrieve
$O(\log n)$ many leftmost (starting with $p_1$) canonical nodes of $T$
whose leaves spans all edges in $P$ that satisfy
the monotonicity condition, i.e., condition (3) as mentioned earlier, with respect to $\Phi$. This can simply be done by recursively
searching children of a node $v$ violating the monotonicity condition
with respect to $\Phi$. Once satisfying the condition, we already have $O(\log
n)$ internal nodes to report their leaves as $P[1,\lambda]$. Searching children and reporting nodes take $O(\log n)$ time
altogether using $O(n)$ space and $O(n\log n)$ preprocessing time.\\

\noindent\textbf{Computing} \textsc{FirstIntersection} {\bf and} \textsc{LastIntersection}{\bf :} Let $\H$ be the hyperplane intersecting $Q[1,2]$ that is perpendicular to $Q[1,2]$ and is tangent to $B(q_2,\eps)$. Let $\H'$ be the other hyperplane perpendicular to $Q[1,2]$ and tangent to $B(q_2,\eps)$. Since $P[1,\lambda]$ is monotone with respect to the line supporting $Q[1,2]$, we know that $\lambda$ must lie on the $q_2$-side of $\H$. And $P(\alpha)\in P[1,\lambda]$ must be located on the first edge intersecting $\H_2$. We start from $p_1$ and perform an exponential search on the edges of $P[1,\lambda]$ to find the first edge that intersects $\H$. Once the edge is found, we can find $P(\alpha)$ in constant time since each edge of $P$ is longer than $2\eps$ which is the diameter of $B(q_2,\eps)$. 
Using the same method we can find $P(\beta )\in P[1,\lambda]$, if we consider $\H'$ instead of $\H$. If $\lambda$ is on the $q_2$-side of $\H'$, we perform the exponential search on $P[1,\lambda]$ to find $P(\beta)$. If $\lambda$ is on the $q_1$-side of $\H'$ then there is no $P(\beta )\in P[1,\lambda]$ and the algorithm does not require it. The whole process takes $O(\log n)$ time.\\

\noindent\textbf{Computing} \textsc{CylinderIntersection}{\bf:} 
Similar to Gudmundsson and Smid~\cite{gs-fqgt-13}, we construct a balanced binary search tree storing the points $p_1, p_2, \ldots , p_n$ in its leaves (sorted by their indices). At each node  of this tree, we store the convex hull of all points stored in its subtree. Given a query range $P[1,\lambda]$, we can retrieve $O(\log n)$ many canonical nodes of the tree containing convex hulls whose leaves span the whole range. For each convex hull we only need to compute extreme points with respect to the direction vector of the edge $Q[1,2]$. If all extreme points lie inside $C(Q[1,2],\eps)$, then $r=null$, otherwise we consider the first extreme point $P(x)$ of some convex hull which lies outside $C(Q[1,2],\eps)$. Note that $P[1,x]$ crosses one of the two boundaries of $C(Q[1,2],\eps)$. Performing exponential search on $P[1,x]$ will find the first point that lies outside the respective boundary of $C(Q[1,2],\eps)$ for which $P(x)$ is obtained. This structure needs $O(n\log n)$ space and preproccessing time and answers queries in $O(\log^2 n)$ time.
Plugging \algref{queryalg} into the decision algorithm (\algref{directedDecider}), 
we obtain the following theorem:
\begin{theorem}[General Curves] \label{thm:DS1}
	Let $P:[1,n]\rightarrow \Reals^2$ be a polygonal curve. A data structure of $O(n \log n)$ size can be built in $O(n \log n)$ time such that for any query curve $Q:[1,m]\rightarrow \Reals^2$ and a positive constant $\eps< \min(l_P/2,l_Q/(1+\sqrt{2}))$, it can be decided in $O(m\log^2 n)$ time whether $\delta_F(P,Q)\leq \eps$.
\end{theorem}

\begin{proof}
	The correctness of the query algorithm follows from
	\thmref{correctness1}. As we mentioned, the
	space and preprocessing time of the whole data structure is 
	$O(n\log n)$. 
	Using  \algref{queryalg}, the longest $\eps$-prefix can be computed in $O(\log^2 n)$ time
	per segment, and hence the query algorithm runs in $O(m\log^2 n)$ time.
	%
\end{proof}

When $P$ is an $x$-monotone curve, we can handle queries
in a slightly faster query time and also smaller
space. 
In this case, we assume that $\eps$ is given at the preprocessing stage.
The $x$-monotonicity of $P$ allows us to use a different data structure for supporting the \textsc{CylinderIntersection} procedure, since the query time and space of this structure dominates the cost of our entire data structure.

\begin{figure}[htbp]
	\centering
	\includegraphics[width=6cm]{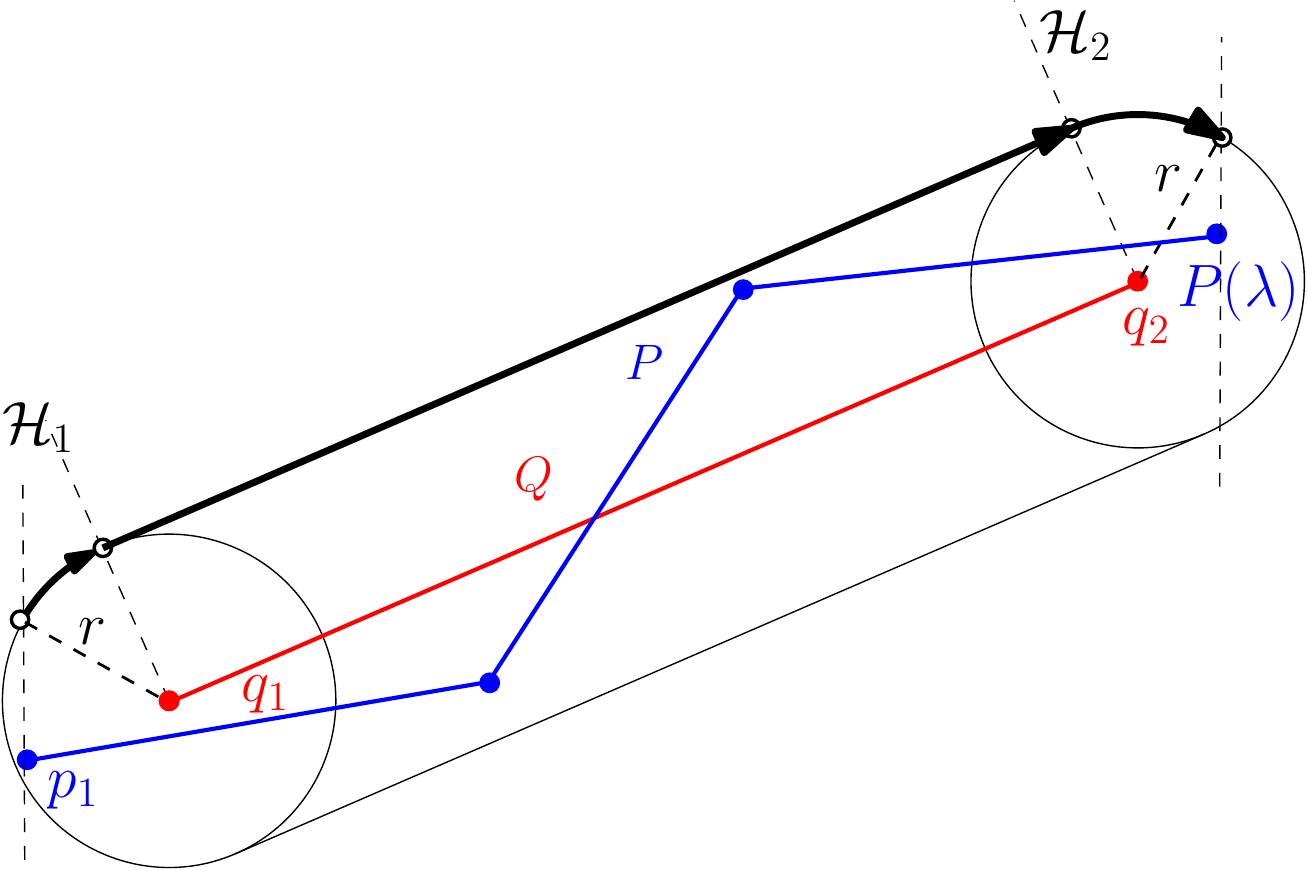}
	\caption{Illustrating how combining the straight and circular line ray shooting queries can find the first point along $P[1,\lambda]$ that leaves $C(Q[1,2],\eps)$.}
	\label{fig:RayShooting}
\end{figure}

We implement the \textsc{CylinderIntersection} procedure by performing two types of ray shooting queries,
straight and circular, along the boundary of $C(Q[1,2],\eps)$. It is easy
to see that it suffices to perform at most two straight ray shooting
queries and four circular ray shooting queries since $P$ is
$x$-monotone.
See Fig.~\ref{fig:RayShooting} for an illustration of the queries for the top part of the boundary of $C(Q[1,2],\eps)$.
%
%

For straight ray shooting queries we use the data structure by
Hershberger and Suri~\cite{hs-parss-95}. Given a simple polygon, their
structure returns the first point on the boundary of the polygon that is
hit by a query ray $\rho$. It can be built in $O(n \log n)$
time using $O(n)$ space and answer queries in $O(\log n)$
time. However, to be able to use this structure we need to reduce our
problem to ray shooting in a simple polygon. Let $P_H$ be the (unbounded) polygon
bounded from below by $P$, 
from the left by a vertical ray from $p_1$ to $\infty$,
and from the right by a vertical ray from $p_n$ to $\infty$. 
Similarly let $P_L$ be the (unbounded) polygon bounded from above by $P$,
from the left by a vertical ray from $p_1$ to $-\infty$,
and from the right by a vertical ray from $p_n$ to $-\infty$. 
We build one data structure
for $P_L$ and one for $P_H$.
For circular ray shooting queries we use the data structure by Cheong
et al.~\cite{cceo-hdcrs-02}. Consider a simple polygon $\P$ with size
$n$ in the plane and let  $r>0$.
For any circular query ray $\rho$ with center $o$, radius $r$, and start point $s'$,
one can report in
$O(\log n)$ query time the first point on the boundary of $\P$ which
is hit by $\rho$.  Combining these structures gives us the first point
along $P[1,\lambda]$ that leaves the cylinder, which completes the implementation of
\textsc{CylinderIntersection}.
We have the following theorem:
\begin{theorem}[$x$-Monotone Preprocessed Curve] \label{cor:DS2}
	Let $\eps>0$ and let $P:[1,n]\rightarrow \Reals^2$ be an $x$-monotone polygonal curve in
	$\Reals^2$ such that $l_P>2\eps$. A linear size data structure can be built
	in $O(n \log n)$ time such that for any polygonal query curve $Q:[1,m]\rightarrow \Reals^2$ with $l_Q>(1+\sqrt{2})\eps$, one can decide in $O(m\log n)$ time whether $\delta_F(P,Q)\leq
	\eps$.
\end{theorem}

\section{Discussion and  Future Work} \label{sec:conclusions}	
In this paper we provided a linear time decision algorithm, an
$O((n+m)\log (n+m))$ time optimization algorithm, a
linear time $\sqrt{d}$-approximation algorithm and a data structure with $O(m\log^2 n)$ query time for the \Frd\ 
between curves that have long edges. Our algorithms are
simple greedy algorithms that run in any constant dimension.
In~\subsecref{assumption}
we gave a critical example that justifies our assumptions on
the edge lengths.  

We proposed several greedy algorithms. Our assumption on the edge lengths allowed us to obtain a linear time constant-factor approximation algorithm for the (continuous) Fr\'echet distance. On the other hand, Bringmann and Mulzer \cite{bw-adfd-15} presented a greedy linear time exponential approximation algorithm for general curves under the discrete Fr\'echet distance. An 
interesting future research direction would be to develop  a trade-off
between the lengths of edges and the runtime, and in general prove hardness
in terms of the edge lengths.

\subparagraph*{Acknowledgements.}
%
We thank the anonymous reviewers for helping to improve the presentation of this paper. We particularly thank one anonymous reviewer for insightful comments that helped us to improve the algorithm in \secref{greedyDecision}.

\clearpage
	
\end{document}